\documentclass[preprint,12pt]{elsarticle}

\usepackage[utf8x]{inputenc}
\usepackage{amsmath,amsfonts,amsthm,amssymb}
\usepackage[english]{babel}
\usepackage{graphicx}
\usepackage{subcaption}
\usepackage{array}
\usepackage{tikz}
\usepackage{comment}
\usepackage[capitalise]{cleveref}
\usetikzlibrary{shapes,snakes}
\tikzset{snake it/.style={decorate, decoration=snake}}

\usepackage[margin=2.3cm]{geometry}

\journal{...}

\DeclareMathOperator{\ecc}{ecc}
\DeclareMathOperator{\rad}{rad}
\DeclareMathOperator{\diam}{diam}

\newcommand\bk[1]{\mathcal{B}_{#1}}

\begin{document}


\begin{frontmatter}

\newtheorem{theorem}{Theorem}[section]
\newtheorem{lemma}[theorem]{Lemma}
\newtheorem{proposition}[theorem]{Proposition}
\newtheorem{observation}[theorem]{Observation}
\newtheorem{corollary}[theorem]{Corollary}
\newtheorem{question}[theorem]{Question}
\newtheorem{conjecture}[theorem]{Conjecture}
\newtheorem{problem}[theorem]{Problem}

\crefformat{observation}{Observation~#1}
\Crefformat{observation}{Observation~#1}
\crefformat{question}{Question~#1}
\Crefformat{question}{Question~#1}

\title{Orienting edges to fight fire in graphs}

\author[lip]{Julien Bensmail\fnref{t1}}
\ead{julien.bensmail.phd@gmail.com}
\author[lip]{Nick Brettell\fnref{t2}} 
\ead{nbrettell@gmail.com}

\address[lip]{LIP, UMR 5668 ENS Lyon, CNRS, UCBL, INRIA, Universit\'e de Lyon, France}
	
\fntext[t1]{This author was supported by ANR project STINT under reference ANR-13-BS02-0007, operated by the French National Research Agency (ANR).}
\fntext[t2]{This author was supported by ANR project Heredia under reference ANR-10-JCJC-0204, operated by the French National Research Agency (ANR).}


\begin{abstract}
%
 We investigate a new oriented variant of the Firefighter Problem. 
 In the traditional Firefighter Problem, a fire breaks out at a given vertex of a graph, and at each time interval spreads to neighbouring vertices that have not been protected, while a constant number of vertices are protected at each time interval.
 In the version of the problem considered here, the firefighters are able to orient the edges of the graph before the fire breaks out, but the fire could start at any vertex.
 We consider this problem when played on a graph in one of several graph classes, and give upper and lower bounds on the number of vertices that can be saved.
 In particular, when one firefighter is available at each time interval, and the given graph is a complete graph, or a complete bipartite graph, we present firefighting strategies that are provably optimal.
 We also provide lower bounds on the number of vertices that can be saved as a function of the chromatic number, of the maximum degree, and of the treewidth of a graph.
 For a subcubic graph, we show that the firefighters can save all but two vertices, and this is best possible.
\end{abstract}

\end{frontmatter}

\section{Introduction} \label{intro_section}

The Firefighter Problem was introduced by Hartnell~\cite{H95} in 1995, and can be described as follows. Suppose we are given a graph $G$, and a vertex $v$ of $G$ at which a fire breaks out.  
At each time unit, the fire propagates from each burning vertex to all of its unprotected neighbours. At the end of each time unit, a firefighter is allowed to permanently protect one 
vertex that is not already burning.
Typically, the firefighters' goal is to prevent as many vertices as possible from burning. Following \cite{FM09}, 
${\rm MVS}(G,\{v\};1)$ denotes the maximum number of vertices of $G$ that can be saved, over all strategies.  More generally, when $f \geq 1$ firefighters can protect the graph at each step, and the fire starts at the vertices in $S \subseteq V(G)$, then the maximum number of vertices of $G$ that can be saved is denoted ${\rm MVS}(G,S;f)$.

This problem has gained increasing attention since its introduction; see~\cite{FM09} for a comprehensive survey. Some investigations into directed versions have also been conducted recently~\cite{BHW12, Kong2014}, where 
the fire propagates from a burnt vertex only through its outgoing incident arcs.
As noted in \cite{FM09}, this directed version is, in a sense, at least as difficult as the undirected version, since there exists
an
orientation for any undirected graph in which the fire propagates as in the undirected version.
Moreover, keeping the fire contained to a small set of burnt vertices is generally difficult, even in the undirected version of the problem. As evidence of this statement, we refer the reader to \cite{KM10}, where the decision problem \textsc{Firefighter} (given a graph $G$ and initial burning vertex $v \in V(G)$, is ${\rm MVS}(G,\{v\};1) \geq k$?) is shown to be \textsf{NP}-complete for cubic graphs; and to~\cite{FKMR07}, where the problem is shown to be \textsf{NP}-complete for trees of maximum degree~3.

\medskip

In this paper, we investigate a new 
variant of this problem,
based on the following question: how can we orient the edges of $G$ in order to minimise the number of burnt vertices? 
Such a problem might be viewed, for example, as a model of the spread of information, or a virus, where there is some mechanism that enforces that the flow is only in one direction. 
Alternatively, imagine a system of rivers, where dams and floodgates can be installed to ensure that, in the event of a flood, the flow is in a certain direction; and structures can be built that block the flow completely.

Note that 
the orientation of $G$ is fixed before the first vertex burns, and cannot be modified later. One motivation for such a restriction is that the operation of orienting $G$ could correspond to a complicated real-life task that is too costly to perform on-demand. 
Moreover, if the orientation can be modified at will, the problem becomes very easy; simply ensure that an edge that is incident to one burnt and one unburnt vertex is oriented towards the burnt vertex.
For the same reason, this version of the problem is only interesting when a fire can break out at any vertex (not known beforehand).

Let $G$ be an undirected graph and let $f$ be an integer at least one.
We can view this problem
as a two-player game played on $G$: player~1 is the fire, and their goal is to maximise the number of vertices that burn, while player~2 is the fire brigade, and their goal is to minimise this number.  The game proceeds as follows: player~2 picks an orientation for $G$, then player~1 picks a vertex at which the fire breaks out, then, at each time interval, player~2 picks $f$ vertices to protect, until the fire no longer propagates.
We denote by $\overrightarrow{\beta}(G, f)$ the number of vertices that burn when both players employ an optimal strategy.

Alternatively,
let $\overrightarrow{\beta}(G, v; f)$ denote the minimum number of vertices that burn for a graph $G$ with $f$ firefighters when the fire starts at the vertex $v$ of $G$, taken over all firefighting strategies and all orientations for $G$; then $\overrightarrow{\beta}(G, f)$ is the maximum of $\overrightarrow{\beta}(G, v; f)$ over all vertices $v$ of $G$.
For the sake of simplicity, we will sometimes adopt the following slight abuse of notation: $\overrightarrow{\beta}(G, f \leq k)$ denotes $\overrightarrow{\beta}(G, f)$ for any $f \leq k$.
For an oriented graph $\overrightarrow{G}$, 
we let $\beta(\overrightarrow{G}, f)$ denote the maximum number of vertices that burn using an optimal firefighting strategy using $f$ firefighters, taken over every possible vertex for the fire to break out.
Thus $\overrightarrow{\beta}(G, f)$ is the minimum of $\beta(\overrightarrow{G}, f)$ taken over all orientations $\overrightarrow{G}$ of $G$.
We analogously define $\beta(G, f)$, where $G$ is an undirected graph, as the maximum number of vertices that will burn using an optimal firefighting strategy 
when $G$ is viewed as a directed graph with arcs $\overrightarrow{uv}$ and $\overrightarrow{vu}$ for each edge $uv$ of $G$.
Thus $\beta(G, f) = |V(G)| - \min_{v \in V(G)} {\rm MVS}(G,\{v\};f)$.



\medskip

Since an undirected graph $G$ can be viewed as a directed graph $\overrightarrow{G}$ where, for each edge $uv$ of $G$, there are arcs $\overrightarrow{uv}$ and $\overrightarrow{vu}$ in $\overrightarrow{G}$,
orienting the edges of an undirected graph effectively decreases the outdegree of some (or all) of the vertices of $G$.
Thus,
orienting the edges of a graph is a very strong tool to prevent the fire from propagating too widely. As further evidence
of this claim, observe that if $u$ is a vertex in $G$ with maximum degree $\Delta$, 
then $u$ is a threat to fire containment. But this threat can be easily managed in an orientation $\overrightarrow{G}$ of $G$ by orienting all edges incident to $u$ towards $u$: in such a situation, a fire that breaks out at $u$ will not propagate any further. 

It is not too surprising that the oriented version of the problem swings the balance in favour of the firefighters, but it is perhaps surprising the extent to which it does so.
Suppose one firefighter is available at each time interval.
We show that for a connected graph $G$, at most one vertex burns using an optimal strategy if and only if $G$ contains at most one cycle.
We describe a strategy by which, for any subcubic graph $G$, at most two vertices burn.
We can also guarantee at most two vertices burn using an optimal strategy on a partial $2$-tree $G$.
For graphs with maximum degree~$4$, at most five vertices burn; but this bound may not be sharp.
Consider the decision problem \textsc{OrientedFirefighter}, where the input is a graph $G$, and the question is: ``is $\overrightarrow{\beta}(G,1) \geq k$?''
As a straightforward consequence of our results, this problem is trivial (running in constant time) when restricted to trees, subcubic graphs, or partial $2$-trees. This is in constrast to the problem \textsc{Firefighter}, which is \textsf{NP}-complete when the input is restricted to these graph classes. 

One other interesting aspect of this problem is that the properties of a `good' orientation $\overrightarrow{G}$ (from the firefighters' point of view) 
are different from the usual properties which are considered `good' in an orientation. For example, having an orientation with large diameter and long longest paths is usually desirable; refer, for example, to the investigations in~\cite{CT78}. In the given context, however, we try to find an orientation that avoids such properties.

Much of our focus, in what follows, is proving an upper bound on $\overrightarrow{\beta}(G,f)$ for any $G$ in some class of graphs.
To find such an upper bound $x$, we need only prove the existence of a `good' orientation and strategy by which we can guarantee no more than $x$ vertices burn.
On the other hand, it seems, in general, more difficult to prove lower bounds, where all possible orientations and strategies must be considered.  
However, a trivial lower bound is given by considering the minimum outdegree over all possible orientations of a graph.
Furthermore, it seems easier to obtain tight lower bounds for dense graphs.
For the class of complete graphs, or the class of complete bipartite graphs, we prove sharp lower bounds when one firefighter is available at each time interval.  Thus, the strategies described that meet these bounds are optimal.

\medskip

In what follows, we assume that $G$ is finite and simple, unless otherwise stated.
We also assume that $G$ is connected; if not, we can consider each connected component of $G$ in turn.
We study the parameter $\overrightarrow{\beta}(G,f)$ throughout, 
assuming that the fire starts at a single vertex, $f \geq 1$, and the firefighters' goal is always to save the maximum number of vertices.

After having introduced some useful tools and basic observations in \cref{section:remarks}, we consider several approaches to finding bounds for $\overrightarrow{\beta}$ in \cref{section:basic,section:invariants,section:families,section:beta1}. We start by considering complete graphs and bipartite graphs, in \cref{section:basic}.  We then demonstrate, in \cref{section:invariants}, some relationships between $\overrightarrow{\beta}$ and several graph invariants: namely, chromatic number, arboricity, and the size of a feedback vertex set. We focus on graphs that have bounded treewidth, bounded degree, or are planar, in \cref{section:families}. In Section~\ref{section:beta1}, we give a characterisation of the class of graphs for which $\overrightarrow{\beta}(G,1)=1$, and discuss a characterisation of the class of graphs $G$ with $\overrightarrow{\beta}(G,1)=k$, where $k \geq 2$. 

\medskip

\begin{figure}[!t]
	\centering
	
	\begin{tikzpicture}[inner sep=0.7mm]
		\node[draw, circle, line width=1pt, fill=black](a) at (0,0.5)[label=above:{\scriptsize $1$}]{};	
	
		\node[draw, diamond, line width=1pt](b) at (2,-1)[label=above:{\scriptsize $1'$}]{};
        \node[draw, circle, line width=1pt](c) at (4,-1) {};
		\node[draw, circle, line width=1pt](d) at (2,-2) {};
		\node[draw, circle, line width=1pt](e) at (4,-2) {};
		
		\node[draw, circle, line width=1pt, fill=black](f) at (-2,-1)[label=above:{\scriptsize $2$}]{};	
		\node[draw, diamond, line width=1pt](g) at (-2,-2)[label=below:{\scriptsize $2'$}]{};
		\node[draw, circle, line width=1pt, fill=black](h) at (-4,-1)[label=above:{\scriptsize $3$}]{};	
				
		\draw[-latex,line width=1pt] (a) -- (b);
		
		\draw[-latex,line width=1pt] (b) -- (c);
		\draw[-latex,line width=1pt] (c) -- (e);
		\draw[-latex,line width=1pt] (d) -- (e);
		\draw[-latex,line width=1pt] (c) -- (d);
		\draw[-latex,line width=1pt] (d) -- (b);
		
		\draw[-latex,line width=1pt] (a) -- (f);
		\draw[-latex,line width=1pt] (f) -- (h);
		\draw[-latex,line width=1pt] (f) -- (g);
		\draw[-latex,line width=1pt] (g) -- (h);
	\end{tikzpicture}
	\caption{An example of firefighting in an oriented graph.}
	\label{figure:drawing-exmp}
\end{figure}
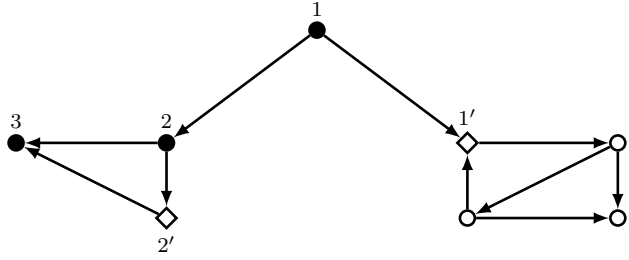

\noindent \textbf{Drawing conventions:} In all figures, a burnt vertex is filled with black, and a label $x$ indicates that this vertex burnt at time~$x$. A diamond vertex represents a protected vertex, with a label $y'$ indicating that this vertex has been protected at time~$y$ (that is, it was protected immediately after the vertices labelled $y$ started burning). Time units are numbered starting from~$1$. See \cref{figure:drawing-exmp} for an example.

\section{Preliminaries} \label{section:remarks}

In this section, we introduce several foundational results that will be used in subsequent sections. We start with the following observation, which will be used to deduce lower bounds on $\overrightarrow{\beta}$.

\begin{observation} \label{observation:subgraph}
Let $H$ be a subgraph of some graph $G$. Then $\overrightarrow{\beta}(G, f) \geq \overrightarrow{\beta}(H, f)$ for any~$f$.
\end{observation}

\begin{proof} 
Let $\overrightarrow{G}$ be any orientation of $G$, and let $\overrightarrow{H}$ be its restriction to $H$. If a fire in $\overrightarrow{G}$ breaks out at some vertex $u \in V(\overrightarrow{G}) \cap V(\overrightarrow{H})$, then at least $\overrightarrow{\beta}(H, f)$ vertices of $\overrightarrow{G}[V(\overrightarrow{H})]$ will burn. 
The inequality follows.
\end{proof}



In the context of the traditional Firefighter Problem, it has been shown that it is often difficult
to prevent the fire from spreading widely.
%
However,
when the firefighters have the ability to orient the graph, firefighting becomes easier,
as there is always an orientation that, essentially, reduces the degree of each vertex by almost a half (we make this precise in \cref{proposition:orientation-half}).
This increases the number of graphs for which firefighting is manageable with a given number of firefighters. 
In particular, firefighting on trees is straightforward, due to the following lemma.


\begin{lemma} \label{lemma:1orienting-trees}
Every tree admits an orientation with maximum outdegree at most~$1$.
\end{lemma}

\begin{proof}
Let $T$ be a tree. Arbitrarily choose a root node $r$ of $T$, and let $\overrightarrow{T}$ be the orientation of $T$ obtained by orienting all edges towards $r$ (that is, if $uv$ is an edge of $T$ and $u$ is nearer to $r$ than $v$, then orient $uv$ from $v$ to $u$). Then $\overrightarrow{T}$ has maximum outdegree at most~$1$.
\end{proof}

\begin{lemma} \label{proposition:orientation-half}
Every graph $G$ admits an orientation $\overrightarrow{G}$ with $$d_{\overrightarrow{G}}^+(v)~\leq~\left\lfloor \frac{d_G(v)}{2} \right\rfloor + 1$$ for each vertex $v \in V(G)$.
\end{lemma}

\begin{proof}
  While $G$ is not a forest, repeatedly pick a cycle $C$ in $G$, add $C$ to a set $\mathcal{C}$, and remove $E(C)$ from $G$. At the end of this procedure, we have a decomposition of $G$ into a forest $F$ and a collection $\mathcal{C}$ of edge-disjoint cycles. The claimed orientation $\overrightarrow{G}$ is obtained by orienting the elements of $F$ and $\mathcal{C}$ as follows:

\begin{itemize}
	\item Orient the edges of every tree $T$ of $F$ so that $T$ has maximum outdegree~$1$. 
          This is possible by \cref{lemma:1orienting-trees}.
	
	\item For every cycle $C$ of $\mathcal{C}$, orient its edges in order to form a directed cycle.
\end{itemize}

Note that orienting any cycle of $\mathcal{C}$ contributes at most~$1$ to the outdegree of each vertex in $\overrightarrow{G}$. Since every vertex $v$ is traversed by at most $\lfloor \frac{d_G(v)}{2} \rfloor$ cycles of $\mathcal{C}$, orienting the cycles of $\mathcal{C}$ contributes at most $\lfloor \frac{d_G(v)}{2} \rfloor$ to the outdegree of $v$. The claim then follows. 
\end{proof}

\medskip


Let $\overrightarrow{G}$ be a directed graph, and let $v$ be a vertex of $\overrightarrow{G}$.  The \emph{eccentricity} of $v$, denoted $\ecc(v)$, is the greatest distance from $v$ to any other vertex of $\overrightarrow{G}$.
The \emph{radius} of $\overrightarrow{G}$, denoted $\rad(\overrightarrow{G})$, is the minimum eccentricity of a vertex of $\overrightarrow{G}$.

We now consider some rough bounds on $\beta$ for oriented graphs with bounded maximum outdegree.

\begin{observation} \label{observation:outdegree}
  Let $\overrightarrow{G}$ be an oriented graph with maximum outdegree~$\Delta^+$. Then,
\begin{enumerate}
  \item[\normalfont{(i)}] $\beta(\overrightarrow{G}, f \geq \Delta^+) = 1$, and
  \item[\normalfont{(ii)}] $\beta(\overrightarrow{G}, 1) \leq |V(\overrightarrow{G})|- \rad(\overrightarrow{G})$.
\end{enumerate}
\end{observation}

	

\begin{proof}
  By positioning $\Delta^+$ firefighters on the outneighbours of the initially burning vertex, (i) is trivial.

  We now consider (ii).
Assume the fire breaks out at $u$, and
partition $V(\overrightarrow{G})$ into layers $\{u\}, V_1, V_2, \dotsc, V_d$, where
$d = \ecc(u)$ and,
for every $i \in \{1, 2, \dotsc, d\}$, the part $V_i$ contains the vertices of $\overrightarrow{G}$ at distance $i$ from $u$.
Note that at time~$i+1$, all vertices in layers $V_1, V_2, \dotsc, V_i$ are either burnt or protected. Now consider the strategy that protects a vertex in $V_i$ at each time unit~$i$. Applying this strategy, at least $d$ vertices will be saved. 
%
In the worst case, when $d$ is at a minimum,
$d=\rad(\overrightarrow{G})$.
The bound then follows.
\end{proof}

We note that when $f = \Delta^+-1$, the basic strategy used in the proof of (ii) is sufficient to  prevent the fire propagating widely.

\begin{corollary} \label{corollary:f-is-delta-1}
For every oriented graph $\overrightarrow{G}$ with maximum outdegree~$\Delta^+$, $$\beta(\overrightarrow{G}, \Delta^+-\nobreak 1) \leq 1+\frac{|V(\overrightarrow{G})|-1}{\Delta^+}.$$
\end{corollary}

\begin{proof}
Applying the strategy described in the proof of \cref{observation:outdegree}, we deduce that at most one new vertex burns at each time unit. 
For each burning vertex $v$, we protect $\Delta^+-1$ outneighbours of $v$, so, in the worst case, $1$ of $\Delta^+$ outneighbours of $v$ burns.  Excluding the vertex at which the fire starts, $\frac{1}{\Delta^+}$ of the $|V(\overrightarrow{G})|-1$ vertices burn, and the result follows.
\end{proof}

We now consider a lower bound for $\overrightarrow{\beta}$ that can be obtained by considering the minimum outdegree of the given graph.

\begin{lemma}
  \label{lowerbound1}
  Let $\overrightarrow{G} = (V,A)$ be a directed graph with maximum outdegree $\Delta^+$.  Then $\Delta^+ \geq \frac{|A|}{|V|}$.
  Moreover, if equality holds, then $d^+(v) = \frac{|A|}{|V|}$ for every $v \in V$.
\end{lemma}
\begin{proof}
  By the handshaking lemma for directed graphs, $$\Delta^+ \cdot |V| \geq \sum_{v \in V} d^+(v) = |A|.$$
  The result follows easily.
\end{proof}

\begin{corollary}
  \label{lowerbound2}
  Let $G$ be a graph.  Then $\overrightarrow{\beta}(G,1) \geq \frac{|E(G)|}{|V(G)|}$.
\end{corollary}

\section{Firefighting in basic graph classes} \label{section:basic}

In this section, we give lower and upper bounds on the number of vertices we can save for trees, complete graphs, and bipartite graphs.

\subsection{Trees}

Unlike for the traditional Firefighter Problem, 
there is an optimal strategy when the given graph is a tree (see \cref{figure:protecting-tree}, for example).

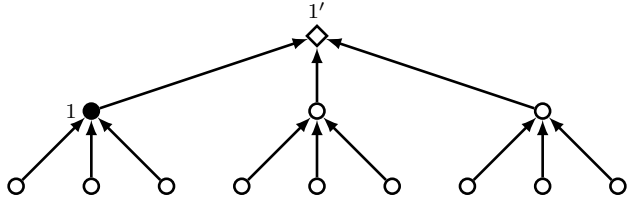
\begin{figure}[!t]
	\centering
	
	\begin{tikzpicture}[inner sep=0.7mm]
		\node[draw, diamond, line width=1pt](r) at (0,0)[label=above:{\scriptsize $1'$}]{};
		
		\node[draw, circle, line width=1pt, fill=black](u1) at (-3,-1)[label=left:{\scriptsize $1$}]{};
		\node[draw, circle, line width=1pt](u2) at (0,-1) {};
		\node[draw, circle, line width=1pt](u3) at (3,-1) {};
		
		\node[draw, circle, line width=1pt](v11) at (-4,-2) {};
		\node[draw, circle, line width=1pt](v12) at (-3,-2) {};
		\node[draw, circle, line width=1pt](v13) at (-2,-2) {};
		
		\node[draw, circle, line width=1pt](v21) at (-1,-2) {};
		\node[draw, circle, line width=1pt](v22) at (0,-2) {};
		\node[draw, circle, line width=1pt](v23) at (1,-2) {};
		
		\node[draw, circle, line width=1pt](v31) at (4,-2) {};
		\node[draw, circle, line width=1pt](v32) at (3,-2) {};
		\node[draw, circle, line width=1pt](v33) at (2,-2) {};
		
		\draw[-latex,line width=1pt] (u1) -- (r);
		\draw[-latex,line width=1pt] (u2) -- (r);
		\draw[-latex,line width=1pt] (u3) -- (r);
		
		\draw[-latex,line width=1pt] (v11) -- (u1);
		\draw[-latex,line width=1pt] (v12) -- (u1);
		\draw[-latex,line width=1pt] (v13) -- (u1);
		
		\draw[-latex,line width=1pt] (v21) -- (u2);
		\draw[-latex,line width=1pt] (v22) -- (u2);
		\draw[-latex,line width=1pt] (v23) -- (u2);
		
		\draw[-latex,line width=1pt] (v31) -- (u3);
		\draw[-latex,line width=1pt] (v32) -- (u3);
		\draw[-latex,line width=1pt] (v33) -- (u3);
	\end{tikzpicture}
	\caption{An optimal orientation for firefighting in a tree.}
	\label{figure:protecting-tree}
\end{figure}

\begin{proposition} \label{proposition:tree}
For every tree $T$, we have $\overrightarrow{\beta}(T, f \geq 1) = 1$.
\end{proposition}

\begin{proof}
By \cref{lemma:1orienting-trees}, every tree admits an orientation with maximum outdegree at most~$1$. The result follows easily. 
\end{proof}


\subsection{Complete graphs}

In this section, we focus on the family of complete graphs. 
We first
present a lower bound on $\overrightarrow{\beta}(K_n, 1)$, 
and then upper bounds on $\overrightarrow{\beta}(K_n, f)$. 
Combining these results, we are able to compute $\overrightarrow{\beta}(K_n, 1)$ for any $n$, demonstrating that the firefighting strategy used to derive the upper bounds is optimal when $f=1$.

The lower bound 
is the following:

\begin{proposition} \label{proposition:lower-complete}
For every $n \geq 1$, we have $\overrightarrow{\beta}(K_n, 1) \geq n-3$. 
\end{proposition}

\begin{proof}
  Clearly, we may assume that $n\geq 4$.
  Let $\overrightarrow{K}$ be an orientation of $K_n$, and let $u$ be a vertex with maximum outdegree $\Delta^+$.  Let $N_1$ be the outneighbours of $u$.
  Since $|E(K_n)| = \frac{n(n-1)}{2}$, it follows, by \cref{lowerbound1}, that $|N_1| \geq \frac{n-1}{2}$, so $|N_1| \geq 2$.
  Let $N_2$ be the vertices in the second outneighbourhood of $u$; that is, $N_2$ contains those vertices not in $\{u\} \cup N_1$ with an incoming incident arc from a vertex of $N_1$.
  Finally, let $N_3 = V(\overrightarrow{K}) \setminus (\{u\} \cup N_1 \cup N_2)$.
  Suppose $N_3$ is non-empty, and
  consider the arcs incident with a vertex $x$ in $N_3$.  Such arcs that are incident with $u$ or a vertex in $N_1$ are oriented away from $x$, since otherwise $x$ would be in the first or second outneighbourhood of $u$.
  Thus $d^+(x) \geq |N_1| + 1 > d^+(u)$; a contradiction.  So $N_3 = \emptyset$.

  Suppose that a fire starts at $u$.
  If the first firefighter is positioned at a vertex in $N_2$, then it follows that $n-2$ vertices will burn.  So we may now assume that a vertex in $N_1$, say $v$, is protected at time~$1$.
  Let $N_1' = N_1 \setminus \{v\}$, let $N_2'$ be the subset of $N_2$ consisting of outneighbours of a vertex in $N_1'$, and let $N_3' = V(\overrightarrow{K}) \setminus (\{u,v\} \cup N_1' \cup N_2')$.
  Since all but at most one of the vertices in $N_2'$ burn at time~$3$, we may assume that $N_3'$ is non-empty (otherwise at least $n-2$ vertices burn).
  Now consider the arcs incident with a vertex $x$ in $N_3'$. Evidently, such arcs that are incident with a vertex in $N_1' \cup \{u\}$ are oriented away from $x$.  Thus $d^+(x) \geq |N_1'| + 1 = d^+(u) = \Delta^+$, so all arcs incident with $x$ and a vertex in $N_2'$ are oriented towards $x$.  
This situation is illustrated in \cref{fig:completeproof}.
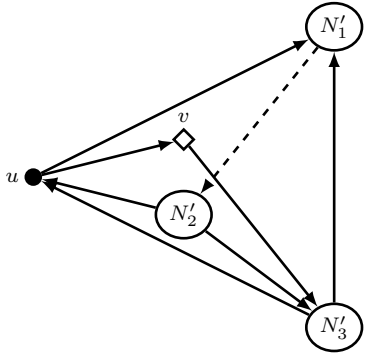
\begin{figure}
    \centering
    \begin{tikzpicture}[inner sep=0.7mm]
        \node[draw, circle, line width=1pt, fill=black](u) at (0,0)[label=left:{\scriptsize $u$}]{};
        \node[draw, diamond, line width=1pt](v) at (2,0.5)[label=above:{\scriptsize $v$}]{};
        \node[draw, ellipse, line width=1pt](w) at (2,-0.5) {\scriptsize $N_2'$};
        \node[draw, ellipse, line width=1pt](uout) at (4,2) {\scriptsize $N_1'$};
        \node[draw, ellipse, line width=1pt](uin) at (4,-2) {\scriptsize $N_3'$};
        
        \draw[-latex,line width=1pt] (u) -- (v);
        \draw[-latex,line width=1pt] (w) -- (u);
        \draw[-latex,line width=1pt] (u) -- (uout);
        \draw[-latex,line width=1pt] (uin) -- (u);
        \draw[-latex,line width=1pt] (uin) -- (uout);
        \draw[-latex,line width=1pt] (w) -- (uin);
        \draw[-latex,line width=1pt] (v) -- (uin);
        \draw[-latex,line width=1pt, dashed] (uout) -- (w);
    \end{tikzpicture}
    \caption{A partial orientation of $K_n$, as in the proof of \cref{proposition:lower-complete}. A solid arrow signifies the direction of all arcs between vertices in the two subsets, whereas a dashed arrow indicates the existence of an arc with the given orientation between the two subsets. }
    \label{fig:completeproof}
\end{figure}

Now, if $|N_2'| > 1$, then a fire starting at $u$ will spread to all unprotected vertices in $N_3'$ at time~$4$.  In this case, $n-3$ vertices burn.  If $N_2' = \emptyset$, then $N_3' = \emptyset$, and $n-1$ vertices burn.
So we may assume that $N_2' = \{w\}$.
Then, if a fire instead starts at $w$, and $|N_3'| > 1$, it spreads to all unprotected vertices in $N_1' \cup \{u\}$ at time~$3$, so at most three vertices can be saved.
In the remaining case, $|N_3'| \leq 1$, so $|N_1'| \geq n-4$, in which case when a fire starts at $u$, at least $n-3$ vertices burn.  This completes the proof.
\end{proof}

\Cref{proposition:lower-complete} is our first confirmation that $\overrightarrow{\beta}$ is not bounded above by some constant for all graphs. 
In particular, for any $k$ there exists a graph $G$ with $\overrightarrow{\beta}(G,1) > k$.

Now we consider
upper bounds on $\overrightarrow{\beta}$ for complete graphs.
First, we focus on complete graphs with odd order, since they admit a regular orientation that facilitates an effective defence strategy.
We then use this result to derive a similar upper bound for complete graphs with even order.

\begin{proposition} \label{observation:complete-upper}
For odd $n \geq 3$, 

	
	

\nopagebreak
\begin{equation*}
  \overrightarrow{\beta}(K_n, f) \leq
  \begin{cases}
    n - 3f & \textrm{if } f < \frac{n-1}{4}, \\
    \frac{n-1}{2} - f + 1 & \textrm{if } \frac{n-1}{4} \leq f < \frac{n-1}{2}, \\
    1 & \textrm{if } f \geq \frac{n-1}{2}. \\
  \end{cases}
\end{equation*}
\end{proposition}

\begin{proof}
Let $V(K_n) = \{v_0, v_1, \dotsc, v_{n-1}\}$ and let $\overrightarrow{K}$ be the orientation of $K_n$ where, for every $i \in \{0, 1, \dotsc, n-1\}$ and $j' \in \{i+1, i+2, \dotsc, i+\frac{n-1}{2}\}$, the edge $v_iv_j$ is oriented from $v_i$ to $v_j$, where $j = j' \mod n$.
Note that $\overrightarrow{K}$ is $\frac{n-1}{2}$-outregular. 
For each vertex $v_i \in V(\overrightarrow{K})$, we associate an ordering $v_{i+1}, v_{i+2}, \dotsc, v_{i+\frac{n-1}{2}}$ on the outneighbours of $v_i$, where the subscripts are interpreted modulo $n$, and when we refer to \emph{consecutive} outneighbours, or outneighbours with the largest indices, we mean with respect to this ordering.

We may assume, by symmetry, that the fire breaks out at $v_0$, and that $f < \frac{n-1}{2}$ (otherwise the fire can be stopped at time~$1$). Let 
$F_1$ be the
$f$ consecutive outneighbours of $v_0$ with the largest indices, and set $B_2 = N^+(v_0) \setminus F_1$. In particular, $|B_2|=\frac{n-1}{2}-f$. Then, at time~$1$, we protect all vertices in $F_1$. By our choice of $F_1$ and $B_2$, the fire will propagate to $B_2$ at time~$2$.

Now let $B_3' = N^+(v_{\frac{n-1}{2}-f}) \setminus F_1$. In other words, $B_3'$ contains those vertices which may potentially burn at time~$3$. Obviously, if $f \geq |B_3'|$, then we can entirely protect $B_3'$ at time~ $2$, and hence stop the fire propagation. The upper bound 
given when $f \geq \frac{n-1}{4}$
then follows.
Now we may assume that $f < \frac{n-1}{4}$.  
Let $F_2$ be the $f$ consecutive vertices of $B_3'$ with the largest indexes, and set $B_3 = B_3' \setminus F_2$. By the remark above, $B_3$ is non-empty and, more precisely, $|B_3|=|N^+(v_{\frac{n-1}{2}-f})|-2f$. We protect the vertices in $F_2$ at time~$2$. The fire then propagates to $B_3$ at time~$3$. Now note that the last vertex of $B_3$ has an outgoing arc towards all unburnt and unprotected vertices
(since $B_3 = \{v_{\frac{n-1}{2}+1}, \dotsc, v_{n-1-2f}\}$ with $2f < \frac{n-1}{2}$).
Let $B_4'$ be this subset of vertices. We have $B_4' = N^-(v_0) \setminus B_3 \setminus F_2$, hence $$|B_4'|=\frac{n-1}{2} - \left(\frac{n-1}{2} - 2f\right) - f = f;$$ so all vertices of $B_4'$ can be protected at time~$3$. 
Thus, the set of vertices that burn is $\{v_0\} \cup B_2 \cup B_3$.  It follows that
$$\beta(\overrightarrow{K}, f) \leq 1 + \left(\frac{n-1}{2} - f \right) + \left(\frac{n-1}{2} - 2f \right) = n - 3f,$$ as claimed.
\end{proof}

Complete graphs with even order do not admit a regular orientation like the one described in the proof of \cref{observation:complete-upper}. However, we can obtain similar bounds for these graphs by `sacrificing' a vertex.


\begin{corollary}
  \label{corollary:complete}
For all even $n \geq 4$, 

\begin{equation*}
  \overrightarrow{\beta}(K_n, f) \leq
  \begin{cases}
    n - 3f & \textrm{if } f < \frac{n-2}{4}, \\
    \frac{n}{2} - f + 1 & \textrm{if } \frac{n-2}{4} \leq f < \frac{n-2}{2}, \\
    2 & \textrm{if } \frac{n-2}{2} \leq f < \frac{n}{2}, \\
    1 & \textrm{if } f \geq \frac{n}{2}. \\
  \end{cases}
\end{equation*}
\end{corollary}

\begin{proof}
  Let $K_n$ be a complete graph, with $n$ even and at least $4$, containing a vertex $v$. Let $\overrightarrow{K}$ be an orientation of $K_n$ for which $\overrightarrow{K} - \{v\}$ is outregular, as in the proof of \cref{observation:complete-upper}, 
  and all arcs incident to $v$ are oriented towards $v$. Then, if the fire breaks out at $v$ in $\overrightarrow{K}$, it will not propagate to any other vertices. If the fire breaks out at some other vertex, then the strategy described in \cref{observation:complete-upper} 
  applies: the
  only difference is that $v$ will also burn. 
  Thus $\overrightarrow{\beta}(K_n, f) \leq \overrightarrow{\beta}(K_{n-1}, f)+1$, so the \lcnamecref{corollary:complete} holds when $f < \frac{n}{2}$.
  Finally, clearly $\overrightarrow{K}$ has outdegree $\frac{n}{2}$ -- so $\overrightarrow{\beta}(K_n, f \geq \frac{n}{2}) = 1$ as required.
\end{proof}

By combining \cref{proposition:lower-complete,observation:complete-upper,corollary:complete},
we deduce the following when $f=1$:

\begin{theorem} \label{theorem:complete-1ff}
For all $n \geq 5$, $$\overrightarrow{\beta}(K_n,1) = n-3.$$
\end{theorem}
On the other hand, $\overrightarrow{\beta}(K_3,1) = 1$ and $\overrightarrow{\beta}(K_4,1) = 2$ (by \cref{lowerbound2}, and \cref{observation:complete-upper} or \cref{corollary:complete} respectively). From \cref{observation:subgraph} and Theorem~\ref{theorem:complete-1ff}, this also gives a lower bound on $\overrightarrow{\beta}(G,1)$ whenever the clique number of $G$ is known.

We suspect that the strategy presented in the proof of \cref{observation:complete-upper} is also optimal when $f > 1$, 
leading to the following conjecture:

\begin{conjecture}
For each $f \geq 1$ and $n > 4f+1$, $$\overrightarrow{\beta}(K_n,f) = n-3f.$$
\end{conjecture}

\subsection{Bipartite graphs}

In this section we consider bounds on $\overrightarrow{\beta}$ for bipartite graphs.
Since bipartite graphs have no cliques of size bigger than two,
\cref{proposition:lower-complete} gives only a trivial lower bound
 on $\overrightarrow{\beta}$ for these graphs.

 We first give a lower bound on $\overrightarrow{\beta}$ for complete bipartite graphs, by finding a lower bound on the maximum outdegree for any orientation of such a graph.

\begin{proposition} \label{proposition:lower-bipartite}
  For positive integers $p$ and $q$, we have $\overrightarrow{\beta}(K_{p,q}, f) \geq \frac{pq}{p+q} + 1 - f$.
  Moreover, when $f \leq \frac{pq}{p+q} -1$, we have $\overrightarrow{\beta}(K_{p,q}, f) \geq \frac{pq}{p+q} + 2 - f$.
\end{proposition}

\begin{proof}
  Let $\overrightarrow{K}$ be an orientation of $K_{p,q}$ with maximum outdegree $\Delta^+$.
  Since
  $\overrightarrow{K}$ has $p+q$ vertices and $pq$ arcs,
  $\Delta^+ \geq \frac{pq}{p+q}$, by \cref{lowerbound1}.
Hence, there exists some vertex $v \in V(\overrightarrow{K})$ such that $d^+(v) \geq \frac{pq}{p+q}$.  If the fire breaks out at $v$, then at least $\frac{pq}{p+q}-f$ vertices will burn at time~$2$, thus proving the first statement of the \lcnamecref{proposition:lower-bipartite}.

  If $\Delta^+ > \frac{pq}{p+q}$, then $\overrightarrow{\beta}(K_{p,q}, f) > \frac{pq}{p+q}+1-f$, satisfying the final statement of the \lcnamecref{proposition:lower-bipartite}. 
  Otherwise, $\Delta^+ = \frac{pq}{p+q}$, so, by \cref{lowerbound1}, every vertex has outdegree precisely $\frac{pq}{p+q}$.  Each of the (at least $\frac{pq}{p+q}-f$) vertices that burn at time 2 has $\frac{pq}{p+q}$ outneighbours, none of which are protected or burning prior to the arrival of the time-$2$ firefighters.  Thus, at least one such vertex burns at time $3$, provided $\frac{pq}{p+q} -f \geq 1$.
%
%
%
\end{proof}


Consider now when $f=1$.  Assume, without loss of generality, that $q \geq p$.  If $q > p(p-1)$, then \cref{proposition:lower-bipartite} implies that $\overrightarrow{\beta}(K_{p,q},1) \geq p$.  We will see, in \cref{observation:upper-bipartite}, that, for such $p$ and $q$, this bound is sharp.
However, when $q$ is much smaller than $p^2$, this bound is poor.
We now consider an improved bound when $p,q \geq 6$.
The proof is similar to that for \cref{proposition:lower-complete}, but requires a more careful case analysis.

\begin{proposition} \label{bipartiteprop}
  Let $K_{p,q}$ be a complete bipartite graph with $p,q \geq 6$.
  Then $$\overrightarrow{\beta}(K_{p,q}, 1) \geq \min\{p,q\}.$$
\end{proposition}

\begin{proof}
  Let $(P,Q)$ be the bipartition of $K_{p,q}$ with $|P|=p$ and $|Q|=q$.
  Let $\overrightarrow{K}$ be an orientation of $K_{p,q}$ and let $u$ be a vertex with maximum outdegree.  Without loss of generality, let $u$ be in $P$.
  We now consider the $i$th outneighbourhood $N_i$ of $u$ in $\overrightarrow{K}$, for each $i$.
  Let $N_1$ be the set of outneighbours of $u$, so each arc incident with $u$ is oriented away from $u$ if and only if its other end is in $N_1$. 
  Let $N_2$ be the subset of $P \setminus \{u\}$ consisting of vertices with an incoming arc from a vertex in $N_1$.
  Let $N_3$ be the subset of $Q \setminus N_1$ consisting of vertices with an incoming arc from a vertex in $N_2$.
  Every vertex $v$ in $P \setminus (N_2 \cup \{u\})$ has arcs towards each vertex of $N_1$, otherwise $v$ would be in $N_2$.
  Since $u$ has maximum outdegree, all other arcs incident with $v$ are oriented towards $v$.  So let $N_4 = P \setminus (N_2 \cup \{u\})$ and observe that all arcs between a vertex in $Q \setminus N_1$ and a vertex in $N_4$ are oriented towards the vertex in $N_4$.  
  Suppose $Q \setminus (N_1 \cup N_3)$ is non-empty, and let $v$ be a vertex in this set.  Then $v$ has every vertex in $P$ as an outneighbour, so a fire starting at $v$ will burn at least $p$ vertices, satisfying the \lcnamecref{bipartiteprop}.  So we may assume that $Q=N_1 \cup N_3$.
%
  This situation is illustrated in \cref{bipartite-layering}.

  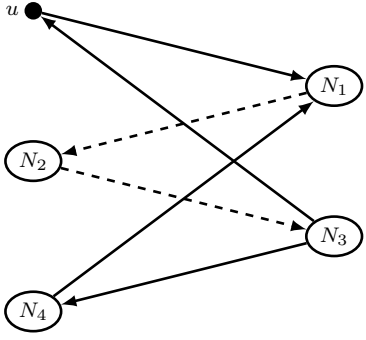
\begin{figure}
    \centering
    \begin{tikzpicture}[inner sep=0.7mm]
        \node[draw, circle, line width=1pt, fill=black](u) at (0,2)[label=left:{\scriptsize $u$}]{};
        \node[draw, ellipse, line width=1pt](n1) at (4,1) {\scriptsize $N_1$};
        \node[draw, ellipse, line width=1pt](n2) at (0,0) {\scriptsize $N_2$};
        \node[draw, ellipse, line width=1pt](n4) at (0,-2) {\scriptsize $N_4$};
        \node[draw, ellipse, line width=1pt](n3) at (4,-1) {\scriptsize $N_3$};
        
        \draw[-latex,line width=1pt] (u) -- (n1);
        \draw[-latex,line width=1pt, dashed] (n1) -- (n2);
        \draw[-latex,line width=1pt, dashed] (n2) -- (n3);
        \draw[-latex,line width=1pt] (n3) -- (u);
        \draw[-latex,line width=1pt] (n3) -- (n4);
        \draw[-latex,line width=1pt] (n4) -- (n1);
    \end{tikzpicture}
    \caption{A partial orientation of $K_{p,q}$, as in the proof of \cref{bipartiteprop}. A solid arrow signifies the direction of all arcs between vertices in the two subsets, whereas a dashed arrow indicates the existence of an arc with the given orientation between the two subsets. }
    \label{bipartite-layering}
  \end{figure}

  We may assume that $N_3 \neq \emptyset$, otherwise if a fire starts at $u$, then $q$ vertices will be burning at time~$2$, satisfying the \lcnamecref{bipartiteprop}.
  Since $p,q \geq 6$, \cref{lowerbound1} implies that $|N_1| \geq 3$.
  If $|N_2| \leq 2$, then $N_4 \neq \emptyset$, and a fire starting at a vertex in $N_3$, say $w$, will spread to all but at most one vertex of $N_4 \cup \{u,w\}$ at time~$2$, and all but at most two vertices of $N_1 \cup N_4 \cup \{u,w\}$ at time~$3$.  Since $|N_1| \geq 3$, at least $p$ vertices burn, as required.  So we may assume that $|N_2| \geq 3$.

  We now deduce further structure by considering when the fire starts at $u$.
  In what follows, when we say that $(X,Y)$ is a \emph{partition} of a set $Z$, the sets $X$ and $Y$ need not be non-empty.
  Let $(N_1',F_1)$ be a partition of $N_1$, 
  let $N_2''$ be the set of outneighbours of $N_1'$,
  and let $(N_2',F_2)$ be a partition of 
  $N_2''$.
  Note that $N_2'' \subseteq N_2$.
  Let $N_3''$ be the set of outneighbours of $N_2'$ in $N_3$, and let $(N_3',F_3)$ be a partition of $N_3''$.
  Also, let $Z = N_3  \setminus N_3''$, so each arc between $Z$ and $N_2'$ is towards $N_2'$.
  Finally, let $N_4''$ be the set of outneighbours of $N_3'$ in $P \setminus (\{u\} \cup N_2'')$, let $Y$ be the remaining vertices in $P$, 
 and let $(N_4', F_4)$ be a partition of $N_4''$.
  We illustrate this situation in \cref{bipartite-layering2}.
  The vertices in $F_1 \cup F_2 \cup F_3 \cup F_4$ represent vertices that are protected in the first $4$ time units if the fire starts at $u$.  So $|F_1 \cup \dotsm \cup F_i| \leq i$ for $i \in \{1,2,3,4\}$.

  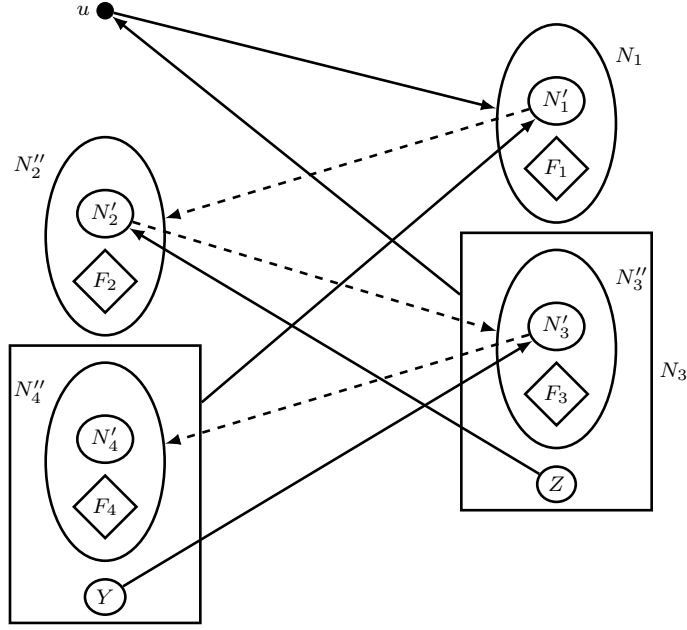
\begin{figure}
    \centering
    \begin{tikzpicture}[inner sep=0.7mm, line width=1pt, scale=1.5]
        \node[draw, circle, fill=black](u) at (0,2)[label=left:{\scriptsize $u$}]{};
        \node[draw, ellipse,minimum width=45pt,minimum height=75pt](n1) at (4,1)[label=above right:{\scriptsize $N_1$}]{};
        \node[draw, ellipse](nn1) at (4,1.2) {\scriptsize $N_1'$};
        \node[draw, diamond](f1) at (4,0.6) {\scriptsize $F_1$};
        \node[draw, ellipse,minimum width=45pt,minimum height=75pt](n2) at (0,0)[label=above left:{\scriptsize $N_2''$}] {};
        \node[draw, ellipse](nn2) at (0,0.2) {\scriptsize $N_2'$};
        \node[draw, diamond](f2) at (0,-0.4) {\scriptsize $F_2$};
        \node[draw, ellipse,minimum width=45pt,minimum height=75pt](nf3) at (4,-1)[label=above right:{\scriptsize $N_3''$}] {};
        \node[draw, ellipse](nn3) at (4,-0.8) {\scriptsize $N_3'$};
        \node[draw, diamond](f3) at (4,-1.4) {\scriptsize $F_3$};
        \node[draw, ellipse,minimum width=45pt,minimum height=75pt](nf4) at (0,-2.0)[label=above left:{\scriptsize $N_4''$}] {};
        \node[draw, ellipse](nn4) at (0,-1.8) {\scriptsize $N_4'$};
        \node[draw, diamond](f4) at (0,-2.4) {\scriptsize $F_4$};
        \node[draw, rectangle,minimum width=72pt,minimum height=105pt](n3) at (4,-1.2)[label=right:{\scriptsize $N_3$}] {};
        \node[draw, rectangle,minimum width=72pt,minimum height=105pt](x3) at (0,-2.2) {};
        \node[draw, ellipse](y) at (0,-3.2) {\scriptsize $Y$};
        \node[draw, ellipse](z) at (4,-2.2) {\scriptsize $Z$};
        
        \draw[-latex] (u) -- (n1);
        \draw[-latex, dashed] (nn1) -- (n2);
        \draw[-latex, dashed] (nn2) -- (nf3);
        \draw[-latex] (n3) -- (u);
        \draw[-latex, dashed] (nn3) -- (nf4);
        \draw[-latex] (y) -- (nn3);
        \draw[-latex] (z) -- (nn2);
        \draw[-latex] (x3) -- (nn1);
    \end{tikzpicture}
    \caption{A partial orientation of $K_{p,q}$, as in the proof of \cref{bipartiteprop}, taking into account vertices that are protected in the first four time intervals when the fire starts at $u$.}
    \label{bipartite-layering2}
  \end{figure}

  Since $|N_1| \geq 3$ and $|F_1| \leq 1$, we have $|N_1'| \geq 2$.
  If $|N_1'| = 2$, then $\Delta^+ = 3$, and it follows that $p=q=6$ and $d^+(v)=3$ for each vertex $v$, by \cref{lowerbound1}.  But in this case, it is easily verified that $|N_2'| \geq 2$ and $|N_3'| \geq 1$, implying that at least $p=q=6$ vertices burn, as required, when a fire starts at $u$.  So we may assume that $|N_1'| \geq 3$.

%
  If a fire starts at $u$,
  then at time~$5$ all vertices in $\{u\} \cup N_1' \cup N_2' \cup N_3' \cup N_4'$ burn (where $N_i'$ may be empty, for some $i \in \{2,3,4\}$, but then $N_i' \cup N_{i+1}' \cup \dotsm \cup N_4' = \emptyset$ by definition).
  If $Y=\emptyset$,
  then this is at least $p + |N_1' \cup N_3'| - |F_2 \cup F_4| \geq p$ vertices, satisfying the \lcnamecref{bipartiteprop}.
%
  So we henceforth assume that $Y \neq \emptyset$. 
  In particular, observe that if
  $|N_3'| > |F_1|$, then any vertex in $Y$ has outdegree more than $u$, so $Y = \emptyset$.
So it remains to consider when $|N_3'| \leq |F_1| \leq 1$; 
the remainder of the proof is dedicated to handling 
this case.

\medskip

First, we show that $|F_1 \cup N_2''| \geq 3$.
If $F_1 = \emptyset$, then $N_2'' = N_2$, which has size at least $3$, satifying the claim.  So assume that $F_1 \neq \emptyset$.
Towards a contradiction, suppose that $|N_2''| \leq 1$.
Recall that the only vertices in $P \setminus \{u\}$ with incoming arcs from $N_1'$ are in $N_2''$. 
The vertices in $P \setminus (\{u\} \cup N_2'')$ can be partitioned into $(N_2 \setminus N_2'', N_4)$, where $|N_2 \setminus N_2''| \geq 2$, since $|N_2| \geq 3$.
We first consider a fire that breaks out at the vertex in $F_1$.
Let $P'$ be a subset of $N_2 \setminus N_2''$ of size at least $|N_2 \setminus N_2''|-1$.  Let $Q'$ be the subset of vertices in $N_3$ that have outgoing arcs towards every vertex in $P'$.  
If $Q' = \emptyset$, then if a fire starts at the vertex in $F_1$, it spreads to $P'$ at time~$2$, since if a vertex in $N_2 \setminus N_2''$ is protected, we may assume it is not in $P'$.  At time~$3$, the fire spreads to all unprotected vertices in $N_3$, since $Q'=\emptyset$ implies that every vertex of $N_3$ is reachable from $P'$; and to unprotected vertices in $N_1'$, since $P'$ is non-empty and all arcs between $P'$ and $N_1'$ are towards $N_1'$.
If the first firefighter is positioned in $P$, then at least $q-1 + |P'|$ vertices burn; otherwise, $q-2 + |N_2 \setminus N_2''|$ vertices burn; in either case $q$ vertices burn as required.
%
Now we may assume that $Q' \neq \emptyset$. Consider a fire starting at $q$ in $Q'$.  At time~$2$, all unprotected vertices in $\{u\} \cup P' \cup N_4$ burn, 
where $|\{u\} \cup P' \cup N_4| \geq p-2$.
At time~$3$, all unprotected vertices in $N_1'$ burn, since vertices in $N_1'$ have incoming arcs from both $u$ and the non-empty set $P'$.  
Hence, a total of at least $1 + (p-2) + 3 - 2 \geq p$ vertices are burning at this time.
So we may assume that $|F_1 \cup N_2''| \geq 3$, and, in particular, that $N_2' \neq \emptyset$.

\medskip

If $Z = \emptyset$, then a fire starting at $u$ burns $q - |F_1 \cup F_3| + |\{u\} \cup N_2'|$ vertices.
This value is at least $q$, since 
$N_2' \neq \emptyset$, and 
when $|N_2'| = 1$, then $|F_1 \cup N_2''| \geq 3$ implies that $|F_2| \geq 1$ so $|F_1 \cup F_3| \leq 2$. 
So we may now assume that $Z \neq \emptyset$.

\medskip

Suppose that $|N_3'| = |F_1| =1$.
Recall
that $Y$ is non-empty.
Since any vertex in $Y$ has every vertex in $N_1' \cup N_3'$ as an outneighbour, all arcs between $Z$ and $Y$ are towards $Y$.
Let $(Z_1,Z_2)$ be the partition of $Z$ such that vertices in $Z_1$ are outneighbours of some vertex in $N_4'$, whereas all arcs incident with a vertex in $Z_2$ and a vertex in $N_4'$ are oriented towards the vertex in $N_4'$.
Suppose that $Z_2$ is non-empty, and consider a fire starting at a vertex in $Z_2$.  Then all but at most one vertex of $\{u\} \cup N_2' \cup N_4' \cup Y$ burns at time~$2$, so a total of at least $p-|F_2 \cup F_4|$ vertices are burning at this time.  
At time~$3$, the fire will spread to
the vertex in $N_3'$, if unprotected, since either $Y$ or $N_2'$ is unprotected,
as well as unprotected vertices of $N_1'$, since $|\{u\} \cup Y| \geq 2$.
Hence, 
at least $p- |F_2 \cup F_4| + |N_1'|$ vertices burn.  Since $|N_1'| \geq 3 \geq |F_2 \cup F_4|$, this is at least $p$ vertices, as required. 

So we may assume that $Z_2$ is empty.  Thus $Z=Z_1$, and this set is non-empty.
If $N_4' = \emptyset$, then, 
as in the previous paragraph, at least $p-|F_2 \cup F_4| + |N_1'| \geq p$ vertices burn when a fire starts at a vertex in $Z$.
So assume that $N_4' \neq \emptyset$.
Now, a fire starting at $u$ spreads to $N_1' \cup N_2' \cup N_3' \cup N_4'$ at time $5$, and all unprotected vertices in $Z=Z_1$ at time $6$, so at least $q-|F_1 \cup F_3| + |N_2' \cup N_4'| $ vertices burn.  This value is at least $q$, because $N_2'$ and $N_4'$ are non-empty, and when $|F_1 \cup F_3| = 3$, then $|N_2'| \geq 2$, since $F_2 = \emptyset$.
So the \lcnamecref{bipartiteprop} holds when $|N_3'| = 1$.

\medskip

Suppose that $N_3' = \emptyset$.  
%
Then, by definition, $N_4'' = \emptyset$.
Recall that $Y,Z \neq \emptyset$.
%
Let $z$ be a vertex in $Z$.
If $z$ has arcs towards every vertex in $Y$, then a fire starting at $z$ spreads to at least $p-|F_2|$ vertices by the end of time~$2$, and, since $|Y| \geq 1$, it spreads to unprotected vertices in $N_1'$ at time~$3$; so $p$ vertices burn as required. 
Thus, for each $z$ in $Z$, there exists a vertex $y$ in $Y$ such that there is an arc from $y$ towards $z$.
Moreover, every vertex in $Y$ has at most one outneighbour in $Z$, since $u$ has maximum outdegree.
It follows that $|Y| \geq |Z|$, and that there exists a subset $Y'$ of $Y$ with $|Y'| = |Z|$ such that each $y$ in $Y'$ has precisely one outneighbour in $Z$, and for distinct $y_1,y_2$ in $Y'$, their outneighbours in $Z$ are distinct.
Note, in particular, that each vertex in $Z$ has $|Z|-1$ outneighbours in $Y'$.


Suppose that $|Z| \geq 2$.  We first consider when $|Y| > |Z|$.  Then there exists some vertex $y'$ in $Y \setminus Y'$.  Since $y'$ has at most one outneighbour in $Z$, there is an arc from some vertex $z'$ in $Z$ to $y'$.  So $z'$ has at least $|Z|$ outneighbours in $Y$; we denote these outneighbours by $N_Y^+(z')$.
%
%
Suppose that a fire starts at $z'$.
Then, at time~$2$, all unprotected vertices in $N_2' \cup \{u\} \cup N_Y^+(z')$ burn.
At time~$3$, unprotected vertices in either $N_1' \cup F_1$ (if $u$ was not protected at time~$1$) or $N_1' \cup F_3$ (otherwise) burn.
%
By the end of time~$3$, at least
$|N_1' \cup Z \cup F'| + |N_2'| + |\{u,z'\}| -2 = q- |F| + |N_2'|$ vertices burn, where
$\{F,F'\} = \{F_1,F_3\}$.
Evidently $|N_2'| \geq |F_1|$; and $3 \leq |F_1 \cup F_2| + |N_2'| \leq 3 - |F_3| + |N_2'|$, so $|N_2'| \geq |F_3|$.
Hence at least $q$ vertices burn 
when $|Y| > |Z|$. 
If $|Y|=|Z|$, 
then, by a similar argument, at least 
$|N_1' \cup F| +  |Z| + |N_2'| -1$
vertices burn, where $F \in \{F_1, F_3\}$.  
Since $|Y|=|Z|$ and $|N_1'| \geq 3$, this value is at least $p - |F_2| + |F| + 1$, so we get the desired result when $|F| \geq |F_2|-1$.  
In the exceptional case, $|F_2|=2$, so $F_1=\emptyset$.
But then
$N_2'' = N_2$, so all arcs between vertices in $Y \subseteq N_4$ and $Z \subseteq N_3$ are towards $Y$; a contradiction.

Finally, suppose $|Z|=1$.  Let $Z=\{z\}$ and let $y$ be the vertex in $Y$ that has $z$ as an outneighbour.
A fire starting at $u$ burns all vertices in $\{u\} \cup N_2' \cup (Q \setminus (F_1\cup F_3 \cup \{z\}))$ by the end of time~$2$, so at least $q$ vertices burn when $|N_2'| \geq |F_1 \cup F_3|$.
So we may assume that $|N_2'| < |F_1 \cup F_3|$.
Then
either $|N_2'| = 2$ and $F_2 = \emptyset$; or $|N_2'|=1$, in which case 
$|F_2|=1$.
Thus $|N_2''|=2$, so $|Y| \geq 3$.
%
Suppose there exists $y' \in Y \setminus \{y\}$ such that $z$ is an outneighbour of $y'$.  Since $u$ has maximum outdegree, arcs between $y$ or $y'$ and $F_3$ are oriented away from $F_3$.
Now a fire starting at a vertex in $F_3$ burns $q - |F_1 \cup F_3| +2$ vertices by the end of time~$2$.
If $F_2 \neq \emptyset$, this is at least $q$ burning vertices, as required; otherwise, unprotected vertices in $N_2'$ burn at time~$3$, so again at least $q$ vertices burn as required.
So we may assume that all arcs between $z$ and $Y \setminus \{y\}$ are oriented away from $z$, where $|Y \setminus \{y\}| \geq 2$.
Now, if a fire starts at $z$, then at least $p-|F_2 \cup \{y\}|$ vertices burn by the end of time~$2$, and then the fire spreads to unprotected vertices of $N_1'$ at time~$3$.  So at least $p$ vertices burn.
This completes the proof.
\end{proof}

%
This proposition implies a lower bound on $\overrightarrow{\beta}$ for any graph containing a complete bipartite subgraph.
Moreover, it
shows that any family of graphs containing graphs with arbitrarily large complete bipartite subgraphs has no constant upper bound for $\overrightarrow{\beta}$.

We now consider an upper bound for bipartite graphs.

\begin{observation} \label{observation:upper-bipartite}
For every bipartite graph $G$ with bipartition $(A, B)$, $$\overrightarrow{\beta}(G, f) \leq 1+ \min\{\Delta(A),  \Delta(B)\} - f.$$
\end{observation}

\begin{proof}
Assume that $\Delta(A) \leq \Delta(B)$ and consider an orientation $\overrightarrow{G}$ of $G$ where all arcs are oriented from $A$ to $B$. Note that if the fire breaks out at some vertex of $B$, then it cannot propagate to other vertices of $G$. Now if the fire breaks out at some vertex $u$ in $A$, then, assuming the firefighters protect $f$ outneighbours of $u$ at time~$1$, at most $\Delta(A)-f$ new vertices will burn at time~$2$. However, the fire will not be able to propagate further, so at most $1+\Delta(A)-f$ vertices burn.
\end{proof}

When $f=1$, Propositions~\ref{proposition:lower-bipartite} and \ref{bipartiteprop} imply that the strategy described in the proof of \cref{observation:upper-bipartite} 
is optimal for $K_{p,q}$ with $q > p(p-1)$ or $\min\{p,q\} \geq 6$.

\begin{theorem} 
For all $p,q \geq 6$, and for any $p\geq 1$ and $q > p(p-1)$,
$$\overrightarrow{\beta}(K_{p,q},1) = \min\{p,q\}.$$
\end{theorem}

\noindent In general, however, the strategy in the proof of \cref{observation:upper-bipartite} may not be optimal, even for complete bipartite graphs. 
For example, for $K_{2,2}$, it follows from \cref{observation:upper-bipartite} that $\overrightarrow{\beta}(K_{2,2}, 1) \leq 2$. But $K_{2,2}$ admits an orientation with maximum outdegree~$1$; hence $\overrightarrow{\beta}(K_{2,2}, 1) = 1$. 
More generally, a cyclic orientation can be used on $K_{p,p}$, similar to that used for complete graphs in the proof of \cref{observation:complete-upper}, to ensure that strictly fewer than $p$ vertices burn when $f \geq \frac{p-1}{3}$.
We conjecture the following:

\begin{conjecture}
  For each $f,p,q \geq 1$ with $\min\{p,q\} > 3f+1$, $$\overrightarrow{\beta}(K_{p,q},f) = 1+\min\{p,q\}-f.$$
\end{conjecture}


\section{Firefighting in graphs with particular properties} \label{section:invariants}

In this section, we describe several strategies for deducing upper bounds on $\overrightarrow{\beta}$. 
In each case we obtain these bounds by exploiting the value of some graph invariant.


\subsection{Graph classes with bounded chromatic number} \label{section:chromatic}

Given an undirected graph $G$, a \textit{proper $k$-vertex-colouring of $G$} is a partition $(V_1, V_2, \dotsc, V_k)$ of $V(G)$ such that $V_i$ is a stable set for each $i \in \{1,2,\dotsc,k\}$. The least $k$ such that $G$ has a proper $k$-vertex-colouring is called the \textit{chromatic number} of $G$, and is denoted $\chi(G)$.

In the next proposition we give an upper bound on 
$\overrightarrow{\beta}$%
, given a graph $G$, in terms of the maximum degree and the chromatic number of $G$.

\begin{proposition} \label{proposition:bound-colorable}
For a graph $G$ with maximum degree $\Delta$,

\begin{enumerate}
  \item[\normalfont{(i)}] $\overrightarrow{\beta}(G, f \geq \Delta) = 1$, and
    	
  \item[\normalfont{(ii)}] $\overrightarrow{\beta}(G, f) \leq \Delta^{\chi(G)}$ for $1 \leq f < \Delta$.
\end{enumerate}
\end{proposition}

\begin{proof}
Set $k=\chi(G)$, and let $\pi=(V_1, V_2, \dotsc, V_k)$ be a proper $k$-vertex-colouring of $G$. Let $\overrightarrow{G}$ be the orientation of $G$ obtained by orienting every edge $uv$ towards the vertex which belongs to the part of $\pi$ with the largest index. That is, if $u \in V_i$ and $v \in V_j$ with $i<j$, then orient $uv$ from $u$ to $v$ (or conversely if $i>j$). Note that the longest paths of $\overrightarrow{G}$ have length $k-1$.
The result follows easily.
\end{proof}

The bound given in \cref{proposition:bound-colorable} when $f < \Delta$ is rough:
we can find a bound that is tighter, but less aesthetically pleasing, 
by considering the number of vertices protected at each step, and
utilising the fact that if a vertex burns at time $t \geq 2$, then it has an in-neighbour, so its outdegree is at most $\Delta -1$.

\begin{proposition} \label{proposition:bound-colorable2}
Let $G$ be a graph with maximum degree $\Delta > 2$ and chromatic number $k$.  Then, for $1 \leq f < \Delta$,
$$\overrightarrow{\beta}(G, f) \leq \frac{\Delta(\Delta-1)^{k-1}-2}{\Delta-2} - f \left(\frac{(\Delta-1)^k-\Delta k+2k-1}{(\Delta-2)^2}\right).$$
\end{proposition}

\begin{proof}
\begin{figure}
	\centering
	
	\begin{tikzpicture}[inner sep=0.7mm]
		\node[draw, circle, line width=1pt, fill=black](u) at (0,0)[label=above:{\scriptsize $1$}]{};	
		
		\node[draw, diamond, line width=1pt](v1) at (-2,-1.5)[label=left:{\scriptsize $1'$}]{};	
		\node[draw, circle, line width=1pt, fill=black](v2) at (0,-1.5)[label=left:{\scriptsize $2$}]{};	
		\node[draw, circle, line width=1pt, fill=black](v3) at (2,-1.5)[label=right:{\scriptsize $2$}]{};	
		
		\node[draw, circle, line width=1pt](w1) at (-3,-3) {};
		\node[draw, circle, line width=1pt](w2) at (-2,-3) {};
		
		\node[draw, diamond, line width=1pt](w3) at (-0.75,-3)[label=left:{\scriptsize $2'$}]{};	
		\node[draw, circle, line width=1pt, fill=black](w4) at (0.75,-3)[label=right:{\scriptsize $3$}]{};	
		
		\node[draw, circle, line width=1pt, fill=black](w5) at (2,-3)[label=left:{\scriptsize $3$}]{};	
		\node[draw, circle, line width=1pt, fill=black](w6) at (3,-3)[label=right:{\scriptsize $3$}]{};	
		
		\draw[-latex,line width=1pt] (u) -- (v1);
		\draw[-latex,line width=1pt] (u) -- (v2);
		\draw[-latex,line width=1pt] (u) -- (v3);
		
		\draw[-latex,line width=1pt] (v1) -- (w1);
		\draw[-latex,line width=1pt] (v1) -- (w2);
		\draw[-latex,line width=1pt] (v2) -- (w3);
		\draw[-latex,line width=1pt] (v2) -- (w4);
		\draw[-latex,line width=1pt] (v3) -- (w5);
		\draw[-latex,line width=1pt] (v3) -- (w6);
		
		
		
		
	\end{tikzpicture}
	\caption{Strategy described following the proof of \cref{proposition:bound-colorable} for $\Delta=3$, $k=3$ and $f=1$.}
	\label{figure:drawing-colorable}
\end{figure}
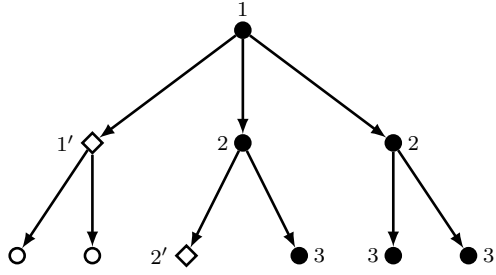

Orient $G$ as described in the proof of \cref{proposition:bound-colorable}.
The maximum number of vertices will burn in the case where the fire starts at a vertex~$v$ with outdegree~$\Delta$, and all the neighbours encountered in a search of depth~$k$ starting at $v$ are distinct (see \cref{figure:drawing-colorable}).
At each time interval, the firefighters protect $f$ outneighbours of burning vertices.
We now calculate the number of vertices that burn in this situation.
Let $S_t$ be the maximum number of vertices that burn at time~$t$. 
Then $S_1=1$, $S_2=\Delta-f$ and, for any $t \geq 3$, we have $S_t = (\Delta-1)S_{t-1} - f$.
By solving this recurrence relation, we deduce that for $t \geq 1$, $$S_{t+1} = \Delta (\Delta-1)^{t-1}-f\left(\frac{(\Delta-1)^t -1}{\Delta-2}\right).$$
The chosen orientation ensures the fire propagates for at most $k$ time intervals.
Thus, an upper bound on the total number of vertices that burn 
is given by $1 + \sum_{t=2}^{k} S_t$. Hence
\begin{align*}
  \overrightarrow{\beta}(G, f) &\leq 1 + \Delta \cdot \sum_{t=1}^{k-1}\left[(\Delta-1)^{t-1}\right] - \frac{f}{\Delta-2} \cdot \sum_{t=1}^{k-1}\left[(\Delta-1)^t-1 \right] \\
  &= 1 +\frac{(\Delta-1)^{k-1}-1}{\Delta-2}\left(\Delta- \frac{f(\Delta-1)}{\Delta-2}\right) + \frac{f(k-1)}{\Delta-2},
\end{align*}
which can be manipulated into the form given in the statement of the \lcnamecref{proposition:bound-colorable2}.
\end{proof}

\Cref{proposition:bound-colorable2} implies, in particular, that, since $f \geq 1$, for any $\Delta > 1$ we have $$\overrightarrow{\beta}(G,f) \leq 2(\Delta-1)^{\chi(G)-1}.$$  Thus,
we can orient the edges of any $3$-colourable graph so that, wherever the fire breaks out, at most $2(\Delta-1)^2$ vertices burn by some firefighting strategy. Furthermore, for a planar graph (or, more generally, a $4$-colourable graph), we can orient its edges so that at most $2(\Delta-1)^3$ vertices burn. By Brooks' Theorem, we have that $$\overrightarrow{\beta}(G, f) \leq 2(\Delta-1)^{\Delta}.$$
In fact, since for a complete graph at most $\Delta$ vertices burn, trivially (as $f \geq 1$), and an odd cycle has an orientation with maximum outdegree~$1$, we have that $$\overrightarrow{\beta}(G, f) \leq 2(\Delta-1)^{\Delta-1}.$$

Thus, any class of graphs with bounded maximum degree has bounded $\overrightarrow{\beta}$.
%
Given a graph $G$ with maximum degree~$3$, 
at most $6$
vertices can burn using an optimal firefighting strategy when one firefighter is available, by \cref{proposition:bound-colorable2}.
For graphs with maximum degree~$4$, the bound is $35$.
%
We will see in \cref{section:subcubic,section:maxdegree4} 
that these bounds are far from best possible.


\subsection{Graph classes with bounded arboricity}

The \textit{arboricity} of an undirected graph $G$, denoted by $a(G)$, is the least number of forests into which the edges of $G$ can be partitioned. 
A graph with small arboricity admits an orientation with small maximum outdegree.

\begin{observation} \label{arboricity}
Every graph $G$ admits an orientation with maximum outdegree at most~$a(G)$.
\end{observation}

\begin{proof}
Let $(E_1, E_2, \dotsc, E_{a(G)})$ be a partition of $E(G)$ inducing forests. Then, for every $i \in \{1, 2, \dotsc, a(G)\}$ and for every tree $T$ of the forest $G[E_i]$, choose an arbitrary orientation of $T$ with maximum outdegree at most~$1$ (which exists by \cref{lemma:1orienting-trees}).
Let $\overrightarrow{G}$ be the orientation of $G$ induced by the orientations of each tree of every $G[E_i]$.
Then, since each vertex $u$ of $G$, in each of the $a(G)$ directed forests, 
has outdegree at most~$1$, $u$ has outdegree at most~$a(G)$ in $\overrightarrow{G}$.
\end{proof}

The following corollary is a straightforward consequence of 
\cref{corollary:f-is-delta-1} and Observations~\ref{observation:outdegree}(i) and~\ref{arboricity}.

\begin{corollary} \label{corollary:arboricity}
For a graph $G$, 

\begin{enumerate}
  \item[\normalfont{(i)}] $\overrightarrow{\beta}(G, f \geq a(G)) = 1$, and
	
  \item[\normalfont{(ii)}] $\overrightarrow{\beta}(G, a(G)-1) \leq 1+\frac{|V(G)|-1}{a(G)}$.
\end{enumerate}
\end{corollary}


\subsection{Graph classes with small feedback vertex set}

A \textit{feedback vertex set} of an undirected graph $G$ is a subset $F \subseteq V(G)$ of vertices whose removal from $G$ results in a forest.
The next observation shows that for a graph with a small feedback vertex set, there is an effective strategy for firefighting.

\begin{observation} \label{observation:feedback}
  Let $F \subseteq V(G)$ be a feedback vertex set of a graph $G$. Then $$\overrightarrow{\beta}(G, f) \leq  \max\{1, |F|-f+2\}.$$
\end{observation}

\begin{proof}
Consider the following orientation $\overrightarrow{G}$ of $G$. First, for every tree of $G - F$, choose a root and orient its edges as described in the proof of \cref{lemma:1orienting-trees} so that $\overrightarrow{G}[V(G) \setminus F]$ has maximum outdegree at most~$1$. Next orient all edges between $V(G) \setminus F$ and $F$ towards $F$. Finally orient all remaining edges, that is those joining vertices in $F$, arbitrarily.

Assume the fire starts at some vertex $u$ of $G$. 
Note that, by the orientation of $\overrightarrow{G}$, the fire cannot propagate from $F$ to $V(G) \setminus F$. 
Moreover, if $u \in V(G) \setminus F$, 
then, using one firefighter, we can stop the propagation of the fire in $G - F$.
Therefore, the worst case is where the fire breaks out at a (non-root) vertex of $V(G) \setminus F$.
In that situation, use the following strategy: at time~$1$, use one firefighter to protect the other vertices of $G - F$ (by protecting the parent of the burnt vertex), and any remaining firefighters to protect vertices of $F$. At time~$2$, all unprotected vertices of $F$ can then burn (if $F$ is complete to $V(G) \setminus F$), but the fire will not be able to propagate further, so $1 + |F| - (f-1)$ vertices burn in this case. 
\end{proof}

\section{Firefighting in particular families of graphs} \label{section:families}

In this section, we give lower and upper bounds on 
$\overrightarrow{\beta}$ for
specific families of graphs.


\subsection{Partial $k$-trees}

A \emph{$k$-tree} is either a complete graph on $k+1$ vertices or a graph that can be obtained from a $k$-tree by adding a vertex that is adjacent to each of $k$ vertices forming a clique.
A \emph{partial $k$-tree} is a subgraph of a $k$-tree.
It is well known that a \textit{$k$-tree} is a maximal graph (in terms of size) with treewidth exactly~$k$, while a \textit{partial $k$-tree} has treewidth at most $k$.

Since every $k$-tree 
contains a clique on $k+1$ vertices, 
we obtain the following lower bound 
using \cref{proposition:lower-complete} and the fact that $\overrightarrow{\beta}(K_4,1) = 2$.

\begin{corollary}
  For a $k$-tree $G$, with $k \geq 3$, we have $\overrightarrow{\beta}(G, 1) \geq \max\{2,k-2\}$.
\end{corollary}

We now give upper bounds on $\overrightarrow{\beta}$ for $k$-trees. 
The proof is based on the existence of an orientation in which the fire can only spread towards a $k$-clique located, loosely speaking, at the centre of the graph.
By definition, a $k$-tree can be constructed starting from \emph{some} $(k+1)$-clique
by repeatedly adding a new vertex that is adjacent to each vertex of a $k$-clique.
We require the following lemma stating that any $k$-tree can be constructed 
in this way starting
from \emph{any} of its $(k+1)$-cliques.
We omit the routine proof.

\begin{lemma} \label{observation:tree-recovery}
Let $G$ be a $k$-tree. For every $(k+1)$-clique $K$ of $G$, we can construct $G$ starting from $K$ by repeatedly adding a vertex that is complete to $k$ vertices forming a clique. 
\end{lemma}

For an undirected graph $G$,
the \emph{diameter} of $G$, denoted $\diam(G)$, is the maximum distance between any two vertices of $G$.

\begin{proposition} \label{proposition:ktree}
For a partial $k$-tree $G$, 

\begin{enumerate}
  \item[\normalfont{(i)}] $\overrightarrow{\beta}(G, f \leq \frac{2k}{\diam(G)}) \leq	1+ \lfloor \frac{\diam(G)}{2} \rfloor \cdot (k-f) - f$, and
	
  \item[\normalfont{(ii)}] $\overrightarrow{\beta}(G, f > \frac{2k}{\diam(G)}) \leq 1+ k \cdot (\lfloor \frac{k}{f} \rfloor - 1)$.
\end{enumerate}
\end{proposition}

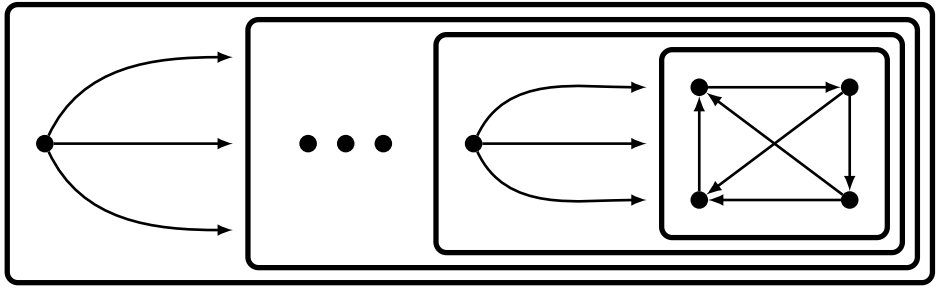
\begin{figure}[!t]
	\centering
	
	\begin{tikzpicture}[inner sep=0.7mm]
		
		\node[draw, circle, line width=1pt, fill=black](u3) at (0, 0) {};
		\node[draw, circle, line width=1pt, fill=black](u4) at (2, 0) {};
		\node[draw, circle, line width=1pt, fill=black](u1) at (0, 1.5) {};
		\node[draw, circle, line width=1pt, fill=black](u2) at (2, 1.5) {};
		
		\draw[-latex, line width=1pt] (u1) -- (u2);
		\draw[-latex, line width=1pt] (u2) -- (u4);
		\draw[-latex, line width=1pt] (u4) -- (u3);
		\draw[-latex, line width=1pt] (u3) -- (u1);
		\draw[-latex, line width=1pt] (u2) -- (u3);
		\draw[-latex, line width=1pt] (u4) -- (u1);
		
		\draw[-,line width=2pt,rounded corners] (-0.5, -0.5) rectangle (2.5,2);
		
		\node[draw, circle, line width=1pt, fill=black](u5) at (-3, 0.75) {};
		
		\draw[-latex, line width=1pt] (u5) to[out=-65,in=180] (-0.7, 0);
		\draw[-latex, line width=1pt] (u5) -- (-0.7, 0.75);
		\draw[-latex, line width=1pt] (u5) to[out=65,in=180] (-0.7, 1.5);
		
		\draw[-,line width=2pt,rounded corners] (-3.5, -0.7) rectangle (2.7,2.2);
		
		\node[draw, circle, line width=1pt, fill=black](u5) at (-4.2, 0.75) {};
		\node[draw, circle, line width=1pt, fill=black](u5) at (-4.7, 0.75) {};
		\node[draw, circle, line width=1pt, fill=black](u5) at (-5.2, 0.75) {};
		
		\draw[-,line width=2pt,rounded corners] (-6, -0.9) rectangle (2.9,2.4);
		
		\node[draw, circle, line width=1pt, fill=black](u6) at (-8.7, 0.75) {};
		
		\draw[-latex, line width=1pt] (u6) to[out=-65,in=180] (-6.2, -0.4);
		\draw[-latex, line width=1pt] (u6) -- (-6.2, 0.75);
		\draw[-latex, line width=1pt] (u6) to[out=65,in=180] (-6.2, 1.9);
		
		\draw[-,line width=2pt,rounded corners] (-9.2, -1.1) rectangle (3.1,2.6);
	\end{tikzpicture}
	\caption{The orientation of a $3$-tree described in the proof of \cref{proposition:ktree}.}
	\label{figure:k-tree}
\end{figure}

\begin{proof}
  By \cref{observation:subgraph}, it suffices to prove these bounds for a (maximal) $k$-tree $G$.
So let $G$ be a $k$-tree, and consider the orientation $\overrightarrow{G}$ of $G$ that we now describe (as illustrated in \cref{figure:k-tree}). Let $K$ be a $(k+1)$-clique in $G$ such that every vertex of $G$ is at distance at most $\lfloor \frac{\diam(G)}{2} \rfloor$ from a vertex of $K$ in $G$. By \cref{observation:tree-recovery}, we can construct $G$ starting from $K$ by repeatedly adding a vertex that is adjacent to each vertex of a $k$-clique in the existing graph. In $\overrightarrow{G}$, first orient the arcs of $K$ as in the proof of \cref{observation:complete-upper}. Then, for each iteration of the construction of $G$ that consists of adding a new vertex $u$ and joining it to all vertices of a $k$-clique, say $K_u$, orient all arcs in $\overrightarrow{G}$ from $u$ towards $V(K_u)$. 


Assume the fire breaks out at some vertex $u$ of $\overrightarrow{G}$. The orientation $\overrightarrow{G}$ of $G$ ensures that the fire can only propagate towards the clique $K$. In particular, the fire will be contained as soon as the vertices of $K$ are reached. Therefore, the most vertices can burn when $u$ is at distance $ \lfloor \frac{\diam(G)}{2} \rfloor$ from a vertex of $K$, so we assume this is the case.

By the construction of $G$ and the orientation of the arcs of $\overrightarrow{G}$, we note that, in $\overrightarrow{G}$, the fire can only propagate to $k$ new vertices at each time unit, and the fire will not reach $K$ until time at most~$\lfloor \frac{\diam(G)}{2} \rfloor$. First, if $f \leq \frac{2k}{\diam(G)}$, then, no matter which vertices are protected at each time unit, the fire spreads to the vertices of $K$ at time $\frac{\diam(G)}{2}$.
In this situation, 
$f$ vertices can be saved at each of the $\frac{\diam(G)}{2}$ time intervals.

On the other hand, if $f > \frac{2k}{\diam(G)}$, then there is a set $S$ of $k$ vertices which will burn at the same time unit at least $\lfloor \frac{k}{f} \rfloor$, when unobstructed by firefighters. 
If, at each time unit, we protect $f$ unprotected vertices of $S$, then all vertices of $S$ will be protected by the time the fire neighbours $S$. In this case, at most $1+ k \cdot (\lfloor \frac{k}{f} \rfloor - 1)$ vertices burn, as required.
\end{proof}

It is worth noting that the anticipation strategy described in the proof of \cref{proposition:ktree} 
demonstrates that, for some oriented $k$-trees, it is not always best to protect vertices adjacent to the fire.
As an example, consider
\textit{$k$th powers of paths}, where the $k$th power $P_n^k$ of the path $P_n$ on $n \geq 1$ vertices is the graph with vertex set $V(P_n)$ for which two vertices are adjacent if and only if they are at distance at most~$k$ in $P_n$.
Using the approach in the proof of \cref{proposition:ktree}, 
pick a $(k+1)$-clique of $P_n^k$ with minimum distance to any other vertex, and then orient all the edges 
towards this centre clique. When $k$ is much greater than $f$ and the underlying path is long, it is clear that if the firefighters protected close to the fire, then the fire would propagate until the fire reaches the centre clique. Hence this is a situation where it is better for the firefighters to anticipate the spread of the fire, as in \cref{proposition:ktree}(ii). 

\medskip

The following is a special case of the strategy described in the proof of \cref{proposition:ktree}.

\begin{proposition} \label{observation:ktree-km1}
For every partial $k$-tree $G$, we have $\overrightarrow{\beta}(G, f \geq \lfloor \frac{k}{2} \rfloor) \leq 1 + \lceil \frac{k}{2} \rceil$.
\end{proposition}

\begin{proof}
By \cref{observation:subgraph}, we may assume that $G$ is a $k$-tree.
Let $\overrightarrow{G}$ be the orientation of $G$ obtained as described in the proof of \cref{proposition:ktree}. Assume the fire breaks out at some vertex $u$ of $\overrightarrow{G}$, and denote by $v_1, v_2, \dotsc, v_k$ its $k$ outneighbours. By the construction of $G$ and the way $\overrightarrow{G}$ was obtained, note that there are $\lceil \frac{k}{2} \rceil$ of the $v_i$'s, say $v_1, v_2, \dotsc, v_{\lceil \frac{k}{2} \rceil}$, with only $\lfloor \frac{k}{2} \rfloor$ outneighbours not among $v_{\lceil \frac{k}{2} \rceil+1}, v_{\lceil \frac{k}{2} \rceil+2}, \dotsc, v_k$. Then protect $v_{\lceil \frac{k}{2} \rceil+1}, v_{\lceil \frac{k}{2} \rceil+2}, \dotsc, v_k$ at time~$1$. The fire will then propagate to $v_1, v_2, \dotsc, v_{\lceil \frac{k}{2} \rceil}$ at time~$2$, but it then suffice to protect $v_{\lceil \frac{k}{2} \rceil+1}, v_{\lceil \frac{k}{2} \rceil+2}, \dotsc, v_k$ at time~$2$ to stop the fire propagation. Note that this strategy remains applicable if $u$ belongs to the root clique since, by its orientation, its vertices have `small' outdegree. Hence, with that strategy, at most $1 + \lceil \frac{k}{2} \rceil$ vertices of $\overrightarrow{G}$ will burn.
\end{proof}

We note that \cref{observation:ktree-km1} is particularly interesting when $k=2$: the edges of every partial $2$-tree can be oriented so that, firefighting with only one firefighter, at most~$2$ vertices burn. This applies to well-known families of partial $2$-trees, such as series-parallel or outerplanar graphs.


\subsection{Subcubic graphs} \label{section:subcubic}

We now focus on \emph{subcubic} graphs: that is, graphs with maximum degree~$3$. Recall that for these graphs, \cref{proposition:bound-colorable2} implies that at most~$6$ vertices burn when firefighting with $1$ firefighter. We reduce this upper bound to~$2$, which is best possible. 

\begin{theorem} \label{subcubic2}
For a subcubic graph $G$, we have $\overrightarrow{\beta}(G, 1) \leq 2$.
\end{theorem}

\begin{proof}
  We will describe an orientation $\overrightarrow{G}$ of $G$ and a firefighting strategy on $\overrightarrow{G}$ for which at most two vertices burn.
  Let $B \subseteq E(G)$ be the set of all bridges of $G$.  We first describe the orientation on this set of edges.
  Note that $B$ induces a forest.  Moreover, the graph $G / (E(G) \setminus B)$, obtained by contracting the edges not in $B$, is a tree.  This tree has an orientation where each edge has outdegree~$1$, by \cref{lemma:1orienting-trees}.
  Let this be the orientation of the edges of $B$ in $\overrightarrow{G}$, and
  call any such arc in $\overrightarrow{G}$ a \emph{bridge arc}.
  After orienting the remaining edges, such an orientation has the property that 
  for each connected (and $2$-connected) component $X$ of $G \backslash B$, there is at most one arc $\overrightarrow{bz}$ in $B$ for which the tail $b$ is in $V(X)$.


  We now consider the orientation of a connected component $X$ of $G \backslash B$ in $\overrightarrow{G}$.  Each vertex of $X$ has degree~$2$ or~$3$ (if there was a vertex of degree~$1$, the incident edge would be a bridge in $G$).
  %
%
  If $X$ consists only of degree-$2$ vertices, then $X$ is a cycle, and we orient the edges such that each edge has indegree~$1$ and outdegree~$1$.
  Now we may assume that there are at least two vertices of degree~$3$ in $X$.
  We will construct a cubic multigraph $X'$, and describe an orientation on $X'$ that extends to $X$.
  We obtain $X'$ by replacing each 
  maximal path $vv_1v_2\dotsm v_pv'$ for which each internal vertex 
  has degree~$2$ with an edge $vv'$
  (see \cref{figure:cubic1}).
  Note that if there is a vertex $b \in V(X)$ for which $B$ has an arc $\overrightarrow{bz}$, then $b$ is contained in some maximal $cc'$-path (say) in $X$ for which each internal vertex has degree~$2$.  
  Clearly, $X'$ is cubic and remains $2$-connected. 
  Thus, according to Petersen's Theorem there exists a partition $(P',C')$ of the edges of $X'$ such that $C'$ induces a collection of cycles (a \emph{$2$-factor}), while $P'$ induces a perfect matching.
  Moreover, it is well-known that there is such a partition for which $P'$ contains any given edge of $X'$.  If $X'$ has an edge $cc'$ corresponding to a path containing $b$, then we pick a partition $(P',C')$ such that $P'$ contains any edge adjacent to $cc'$; that is, $C'$ contains $cc'$.

\begin{figure} 
  \begin{subfigure}{0.5 \textwidth}
        \centering
        \begin{tikzpicture}[inner sep=0.7mm]
                \node[draw, circle, line width=1pt](a0) at (1.75,-1) {};
                \node[draw, circle, line width=1pt](a1) at (0,0) {};
                \node[draw, circle, line width=1pt, fill=black](a2) at (0,1)[label=left:{\scriptsize $v$}]{};
                \node[draw, circle, line width=1pt](a3) at (1,2) {};
                \node[draw, circle, line width=1pt](a4) at (2.5,2) {};
                \node[draw, circle, line width=1pt, fill=black](a5) at (3.5,1)[label=right:{\scriptsize $v'$}]{};
                \node[draw, circle, line width=1pt](a6) at (3.5,0) {};
                                
                \draw[-,line width=1pt] (a0) -- (a1);
                \draw[-,line width=1pt] (a1) -- (a2);
                \draw[-,line width=1pt] (a2) -- (a3);
                \draw[-,line width=1pt] (a3) -- (a4);
                \draw[-,line width=1pt] (a2) -- (a5);
                \draw[-,line width=1pt] (a4) -- (a5);
                \draw[-,line width=1pt] (a5) -- (a6);
                \draw[-,line width=1pt] (a6) -- (a0);
                \draw[-,line width=2pt,rounded corners] (-0.5, -1.5) rectangle (4, 2.5);
                
                \draw[-, dotted, line width=1pt] (a3) -- (1,3);
                \draw[-, dotted, line width=1pt] (a4) -- (2.5,3);
                \draw[-, dotted, line width=1pt] (a0) -- (1.75,-2);
                
                \draw[-, dotted, line width=1pt] (-1,0) -- (a1);
                \draw[-, dotted, line width=1pt] (a6) -- (4.5,0);

                \node at (-1,-1.25) {$X$};
        \end{tikzpicture}
  \end{subfigure}%
  \begin{subfigure}{0.5 \textwidth}
      \centering
        \begin{tikzpicture}[inner sep=0.7mm]
                
                \node[draw, circle, line width=1pt, fill=black](a22) at (5.5,1)[label=left:{\scriptsize $v$}]{};
                \node[draw, circle, line width=1pt, fill=black](a52) at (9,1)[label=right:{\scriptsize $v'$}]{};
                                
                \draw[-,line width=1pt] (a22) -- (a52);
                \draw[-,line width=1pt] (a22) to[out=90,in=90] (a52);
                \draw[-,line width=1pt] (a22) to[out=-90,in=-90] (a52);
                \draw[-,line width=2pt,rounded corners] (5, -0.5) rectangle (9.5, 2.5);

                \node at (10,-.25) {$X'$};
        \end{tikzpicture}
   \end{subfigure}%
        \caption{An example of a subcubic component $X$ of $G \backslash B$, and the corresponding cubic multigraph $X'$, as described in the proof of \cref{subcubic2}. A dotted edge represents a bridge in $B$.}
        \label{figure:cubic1}
\end{figure}
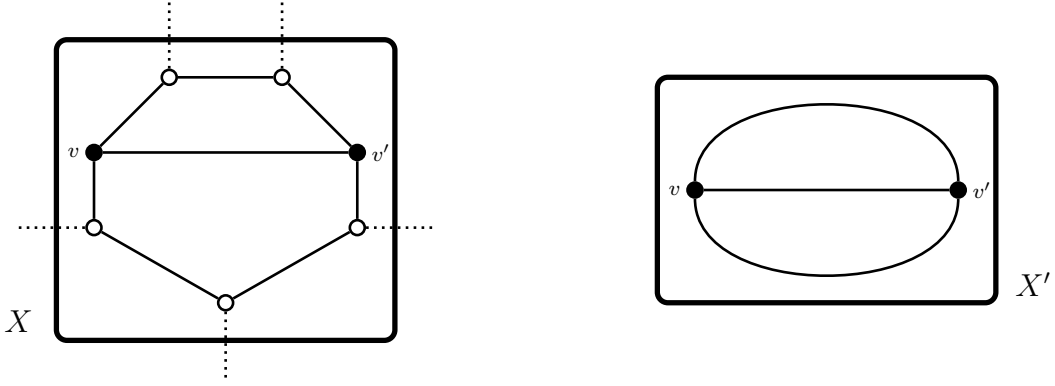

We orient the edges of $X'$ such that
each cycle of $C'$ is $1$-outregular (each vertex in $C'$ has indegree~$1$ and outdegree~$1$), and each edge of $P'$ is oriented arbitrarily. 
This orientation extends to an orientation of the edges in $X$, in the obvious way.  Let $(P,C)$ be the partition of edges of $X$ for which $P$ (respectively, $C$) contains each edge in the path $vv_1\dotsm v_pv'$ corresponding to an edge $vv'$ of $P'$ (respectively, $C'$).
We call an arc in $C$ a \emph{cycle arc}, and an arc in $P$ a \emph{path arc}.
Note that $P$ induces a disjoint union of paths, each oriented from one end to the other.
Moreover, each vertex $v \in V(X)$ is incident to at most two outgoing arcs, and, when $v$ has outdegree precisely two, it has one outgoing path arc, one outgoing cycle arc, and one incoming cycle arc.





\begin{figure}
  \begin{subfigure}{0.5 \textwidth}
        \centering
        
        \begin{tikzpicture}[inner sep=0.7mm]
                \node[draw, circle, line width=1pt](a2) at (0,1) {};
                \node[draw, circle, line width=1pt](a5) at (3.5,1) {};
                                
                \draw[-,line width=2pt] (a2) -- (a5);
                \draw[-,line width=2pt] (a2) to[out=90,in=90] (a5);
                \draw[-,line width=1pt] (a2) to[out=-90,in=-90] (a5);
                \draw[-,line width=2pt,rounded corners] (-0.5, -0.5) rectangle (4, 2.5);
        
                \node at (-1,-0.25) {$X'$};
        \end{tikzpicture}
  \end{subfigure}
  \begin{subfigure}{0.5 \textwidth}
        \centering
        \begin{tikzpicture}[inner sep=0.7mm]
        
                \node[draw, circle, line width=1pt](a1) at (6.5,0){};
                \node[draw, circle, line width=1pt](a2) at (6.5,1){};
                \node[draw, circle, line width=1pt, fill=black](a3) at (7.5,2)[label=left:{\scriptsize $1$}]{};
                \node[draw, diamond, line width=1pt](a4) at (9,2)[label=right:{\scriptsize $1'$}]{};
                \node[draw, circle, line width=1pt](a5) at (10,1) {};
                \node[draw, circle, line width=1pt](a6) at (10,0) {};
                \node[draw, circle, line width=1pt](a7) at (8.25,-1) {};
                                
                \draw[latex-,line width=1pt] (a1) -- (a2);
                \draw[-latex,line width=2pt] (a2) -- (a3);
                \draw[-latex,line width=2pt] (a3) -- (a4);
                \draw[latex-,line width=2pt] (a2) -- (a5);
                \draw[-latex,line width=2pt] (a4) -- (a5);
                \draw[latex-,line width=1pt] (a5) -- (a6);
                \draw[latex-,line width=1pt] (a6) -- (a7);
                \draw[latex-,line width=1pt] (a7) -- (a1);
                \draw[-,line width=2pt,rounded corners] (6, -1.5) rectangle (10.5, 2.5);
                
                \node[draw, diamond, line width=1pt](b1) at (6,4)[label=above:{\scriptsize $2'$}]{};
                \node[draw, circle, line width=1pt, fill=black](b2) at (7.5,4)[label=right:{\scriptsize $2$}]{};
                \node[draw, circle, line width=1pt](b3) at (8,5) {};
                
                \draw[latex-,line width=2pt] (b1) -- (b2);
                \draw[latex-,line width=2pt] (b2) -- (b3);
                \draw[-,line width=2pt,rounded corners] (5.5, 3.5) rectangle (8.5,5.5);
                
                
                \draw[-latex,line width=1pt,dotted] (a3) -- (b2);
                \draw[latex-,line width=1pt,dotted] (a4) -- (10.5,4);
                \draw[latex-,line width=1pt,dotted] (a7) -- (8.25,-2);
                
                \draw[-latex, dotted, line width=1pt] (5.5,0) -- (a1);
                \draw[latex-, dotted, line width=1pt] (a6) -- (11,0);
                
                \node at (11,-1.25) {$X$};
        \end{tikzpicture}
      \end{subfigure}
      \caption{An example of how the orientation $\protect\overrightarrow{G}$ is obtained, and an application of the firefighting strategy on $\protect\overrightarrow{G}$, as described in the proof of \cref{subcubic2}. Thin arcs (respectively, edges) represent path arcs (respectively, edges in the perfect matching $P'$), the thick arcs (respectively, edges) represent cycle arcs (respectively, edges in the $2$-factor $C'$), and the dotted arcs represent bridge arcs.} 
      \label{figure:cubic2}
\end{figure}
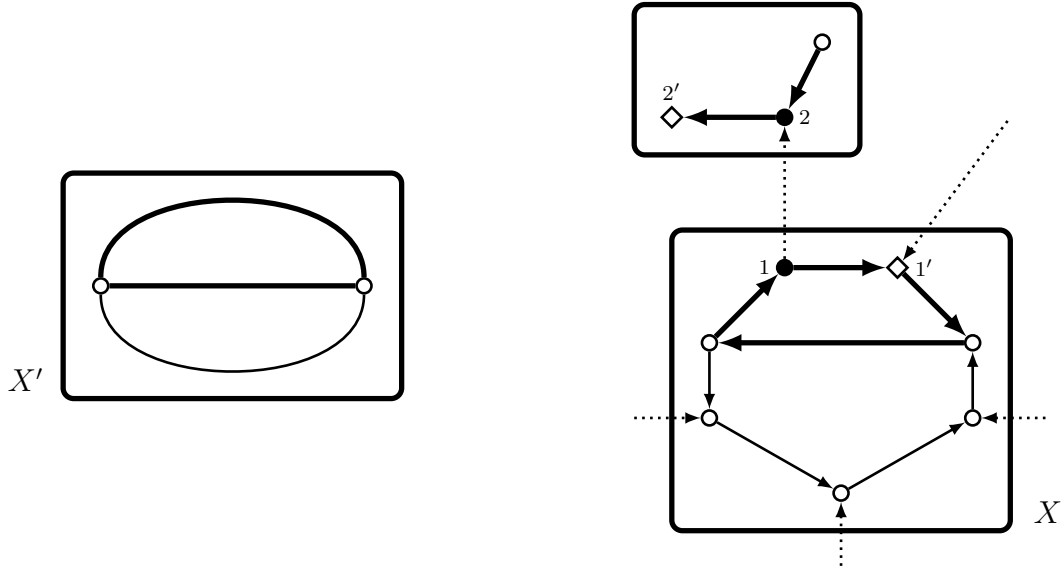

Now consider the orientation $\overrightarrow{G}$ obtained by combining the described orientations on $B$ and each component $X$ of $G \backslash B$ (see \cref{figure:cubic2}). A vertex may be incident to at most one outgoing bridge arc.
However, by the foregoing, such a vertex is either incident to two cycle arcs (one outgoing, one incoming) or other bridge arcs (each incoming).
Thus, the orientation $\overrightarrow{G}$ has the property that every vertex has outdegree at most two, and each vertex with outdegree two is incident to two cycle arcs. 

Finally, we show that, regardless of where the fire breaks out in $\overrightarrow{G}$, there is a strategy, using one firefighter at each step, for which at most two vertices burn.
Say the fire breaks out at some vertex $u$ of $\overrightarrow{G}$. 
If $u$ has outdegree~$1$, containing the fire is trivial.
Otherwise, $u$ has outdegree~$2$, and hence is incident to cycle arcs $\overrightarrow{tu}$ and $\overrightarrow{uv}$, and a bridge arc or path arc $\overrightarrow{uw}$, say. The firefighter blocks $v$ at time~$1$, so,
at time~$2$, the fire spreads to $w$.
Now we prove that $w$ has outdegree~$1$, so the fire can be completely contained.
Evidently this is the case if $w$ is incident to cycle arcs.  If there is a path arc $\overrightarrow{wx}$, the presence of adjacent path arcs implies that $w$ has degree~$2$, and hence outdegree~$1$. Finally, if $w$ is incident only to bridge arcs, then it has outdegree~$1$, by the choice of orientation of the edges $B$ in $\overrightarrow{X}$.  This completes the proof.
%
	%
	%
	%
\end{proof}

Theorem~\ref{subcubic2} is best possible, since there are subcubic graphs, such as $K_4$ or the Petersen graph, for which at least two vertices will necessarily burn, 
by \cref{lowerbound2}.

\subsection{Graphs with bounded maximum degree} \label{section:maxdegree4}

Recall that, so far, the best upper bound on $\overrightarrow{\beta}$ we have seen when firefighting with one firefighter in a graph with maximum degree~$d$ is roughly $2(d-1)^{d-1}$. 
When $d \leq 3$, however, $\overrightarrow{\beta}$ is at most $2$.
In this section, we 
start by considering the case where $d=4$, and show that, for such graphs, $\overrightarrow{\beta}$ is at most $5$.  We then use a similar strategy to improve the upper bound in general, for $d \geq 5$.
%
Here we are interested in the case where $f=1$, 
although a similar approach could also be used to obtain bounds when $f \geq 2$.
We define the following for legibility:
$$\overrightarrow{\beta_d} = \max \{\overrightarrow{\beta}(G, 1) : G {\rm ~is~a~graph~with~maximum~degree~} d\}.$$
So, for example, $\overrightarrow{\beta_3}=2$, by \cref{subcubic2}.

First, we observe that when finding upper bounds on $\overrightarrow{\beta}$ for the class of graphs with maximum degree~$\Delta$, we can restrict our attention to $\Delta$-regular graphs.  

\begin{lemma} \label{observation:reg-subgraph}
If there exists an integer $x$ such that $\overrightarrow{\beta}(G , f) \leq x$ for every $\Delta$-regular graph $G$, then $\overrightarrow{\beta}(G, f) \leq x$ for every graph $G$ with maximum degree $\Delta$.
\end{lemma}
\begin{proof}
  We will show that 
for every graph $G$ with maximum degree~$\Delta$, there exists a $\Delta$-regular graph containing $G$ as a subgraph; the lemma follows from this claim.
  Let $G$ be a graph with minimum degree~$d$ and maximum degree~$\Delta$.
  We describe a construction by which we can obtain a $\Delta$-regular graph 
  that contains $G$ as a subgraph. 
  Clearly the \lcnamecref{observation:reg-subgraph} holds if $d = \Delta$, so
assume that $G$ is not $\Delta$-regular. Take two copies $G_1$ and $G_2$ of $G$, and, for every vertex $v$ of $G$ with degree strictly less than $\Delta$, add an edge between the two vertices corresponding to $v$ in $G_1$ and $G_2$.
We obtain a graph with minimum degree strictly greater than $d$, maximum degree $\Delta$, and containing $G$ as a subgraph.
By repeating this process, for $\Delta-d$ iterations, we eventually obtain a $\Delta$-regular graph as desired.
\end{proof}

\begin{proposition} \label{4reg}
 We have $\overrightarrow{\beta_4} \leq 5$.
\end{proposition}
\begin{proof}
  Let $G$ be a graph with maximum degree~$4$. 
  Let $X$ be the subset of $V(G)$ obtained by the following iterative procedure: starting with $X=\emptyset$, while $G-X$ has a degree-$4$ vertex $x$, add $x$ to $X$ and remove $x$ from $G$.
  Note that when this procedure terminates, $X$ contains no two adjacent vertices of $G$, and $G-X$ has maximum degree at most $3$.
  Let $G' = G-X$.
  Recall that, by \cref{subcubic2}, 
  $\overrightarrow{\beta}(G', 1) \leq 2$.
  Let $\overrightarrow{G'}$ be the orientation of $G'$ described in the proof of \cref{subcubic2}.
  Let $\overrightarrow{G}$ be the orientation of $G'$ such that $\overrightarrow{G} -X= \overrightarrow{G'}$, and all edges incident to a vertex $x$ in $X$ are oriented towards $x$.

  If the fire starts at a vertex in $X$, it cannot propagate any further due to the orientation of $\overrightarrow{G}$.  So suppose the fire starts at a vertex of $G'$.
  We then employ the same firefighting strategy on $G'$ as in the proof of \cref{subcubic2}.
  Whenever the fire spreads to a vertex of $X$, the orientation ensures that it will not spread further from this vertex.  So we only need to consider how many vertices burn in $G'$, and how many vertices of $X$ are adjacent to these burnt vertices in $G$.
%
  Observe that at most two vertices of $G'$ burn. 
  If the two vertices of $G'$ that burn have degree~$3$ in $G'$, then these two vertices each have at most one neighbour in $X$, so at most $4$ vertices burn in total.
  If, instead, the fire starts at a vertex that has degree at most $2$ in $G'$, then this vertex has one in-neighbour and one outneighbour in $\overrightarrow{G'}$.  So only one vertex of $G'$ burns, and hence at most $4$ vertices burn in total.
  Finally, suppose the fire starts a vertex $v_1$ with degree~$3$ in $G'$, then spreads to a vertex $v_2$ at time~$2$ that has degree~$2$ in $G'$.
  Then $v_1$ is adjacent to at most one vertex of $X$ in $G$, and $v_2$ is adjacent to at most two vertices of $X$ in $G$.  So at most $5$ vertices burn in total.  This completes the proof.
\end{proof}

Since $K_{4,4}$ is $4$-regular and $\overrightarrow{\beta}(K_{4,4}, 1)=3$, by \cref{proposition:lower-bipartite}, we have $\overrightarrow{\beta_4} \geq 3$. 
However,
finding its precise value remains an open problem.
\begin{question}
  What is the value of $\overrightarrow{\beta_4}$?
\end{question}

We now focus on improving the upper bound on $\overrightarrow{\beta_d}$ for any $d \geq 5$. 
The next lemma is key to our approach.
The proof is similar to that for \cref{4reg}, but does not rely on properties of the optimal orientation for subcubic graphs.

\begin{lemma} \label{lemma:degree-to-degree}
For every $d \geq 2$, we have $\overrightarrow{\beta_d} \leq \max\{d,\overrightarrow{\beta_{d-1}} \cdot (d-2)+2\}$.
\end{lemma}

\begin{proof}
Let $d \geq 2$ be fixed, and assume $\overrightarrow{\beta_{d-1}} \leq k$ for some $k \geq 1$.
Let $G$ be a graph with maximum degree $d$.
We will describe an orientation $\overrightarrow{G}$ of $G$ on which we can firefight with only one firefighter in such a way at most $\max\{d,\overrightarrow{\beta_{d-1}} \cdot (d-2)+2\}$ vertices burn.

As in the proof of \cref{4reg}, there is a subset $X$ of $V(G)$ that
contains no two adjacent vertices of $G$, and $G-X$ has maximum degree at most $d-1$.
Let $G' = G-X$.  
By the initial assumption, there is an orientation $\overrightarrow{G'}$ of $G'$ with $\beta(\overrightarrow{G'}, 1) \leq k$.
Now, for each edge in $G'$, we give the same orientation to the edge in $\overrightarrow{G}$ as in $\overrightarrow{G'}$; whereas each edge incident to a vertex $x$ of $X$ is oriented towards $x$. 


We now consider a strategy for firefighting on $\overrightarrow{G}$. Assume the fire breaks out at some vertex~$u$. 
If $u$ is in $X$, then it is a sink, and
the fire cannot propagate to the other vertices. 
Now assume $u$ is 
not in $X$.
Due to the way the edges of $G'$ have been oriented in $\overrightarrow{G}$, there is a certain strategy that the firefighter can apply so that at most $\overrightarrow{\beta_d}$ vertices of $G'$ burn. So apply this strategy. Then $x \leq \overrightarrow{\beta_d}$ vertices $u_1, u_2, \dotsc, u_x$ in $G'$ burn, where $u_1 = u$. 

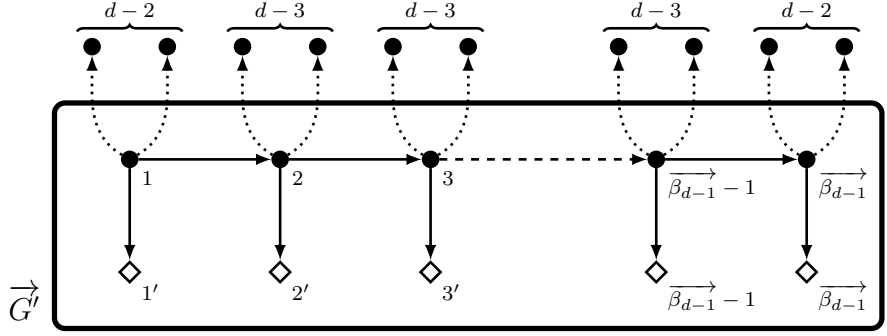
\begin{figure} 
	\centering
	
	\begin{tikzpicture}[inner sep=0.7mm]
		
		\node[draw, circle, line width=1pt, fill=black](u1) at (0, 0)[label=below right:{\scriptsize $1$}]{};
		\node[draw, diamond, line width=1pt](v1) at (0, -1.5)[label=below right:{\scriptsize $1'$}]{};
		
		\node[draw, circle, line width=1pt, fill=black](u2) at (2, 0)[label=below right:{\scriptsize $2$}]{};
		\node[draw, diamond, line width=1pt](v2) at (2, -1.5)[label=below right:{\scriptsize $2'$}]{};
		
		\node[draw, circle, line width=1pt, fill=black](u3) at (4, 0)[label=below right:{\scriptsize $3$}]{};
		\node[draw, diamond, line width=1pt](v3) at (4, -1.5)[label=below right:{\scriptsize $3'$}]{};
		
		\node[draw, circle, line width=1pt, fill=black](ux) at (7, 0)[label=below right:{\scriptsize $\overrightarrow{\beta_{d-1}}-1$}]{};
		\node[draw, diamond, line width=1pt](vx) at (7, -1.5)[label=below right:{\scriptsize $\overrightarrow{\beta_{d-1}}-1$}]{};
		
		\node[draw, circle, line width=1pt, fill=black](ux2) at (9, 0)[label=below right:{\scriptsize $\overrightarrow{\beta_{d-1}}$}]{};
		\node[draw, diamond, line width=1pt](vx2) at (9, -1.5)[label=below right:{\scriptsize $\overrightarrow{\beta_{d-1}}$}]{};
		
		\draw[-latex, line width=1pt] (u1) -- (v1);
		\draw[-latex, line width=1pt] (u1) -- (u2);
		\draw[-latex, line width=1pt] (u2) -- (v2);
		\draw[-latex, line width=1pt] (u2) -- (u3);
		\draw[-latex, line width=1pt] (u3) -- (v3);
		\draw[-latex, line width=1pt, dashed] (u3) -- (ux);
		\draw[-latex, line width=1pt] (ux) -- (vx);
		\draw[-latex, line width=1pt] (ux) -- (ux2);
		\draw[-latex, line width=1pt] (ux2) -- (vx2);
		
		\draw[-,line width=2pt,rounded corners] (-1,0.75) rectangle (10, -2.25);
        \node at (-1.4, -1.9) {$\overrightarrow{G'}$};
        
        \node[draw, circle, line width=1pt, fill=black](w1) at (-0.5, 1.5) { };
        \node[draw, circle, line width=1pt, fill=black](w2) at (0.5, 1.5) { };
		
        \node[draw, circle, line width=1pt, fill=black](w3) at (1.5, 1.5) { };
        \node[draw, circle, line width=1pt, fill=black](w4) at (2.5, 1.5) { };
		
        \node[draw, circle, line width=1pt, fill=black](w5) at (3.5, 1.5) { };
        \node[draw, circle, line width=1pt, fill=black](w6) at (4.5, 1.5) { };
		
        \node[draw, circle, line width=1pt, fill=black](w7) at (6.5, 1.5) { };
        \node[draw, circle, line width=1pt, fill=black](w8) at (7.5, 1.5) { };
		
        \node[draw, circle, line width=1pt, fill=black](w9) at (8.5, 1.5) { };
        \node[draw, circle, line width=1pt, fill=black](w10) at (9.5, 1.5) { };
		
		\draw[-latex, line width=1pt, dotted] (u1) to[out=150,in=-90] (w1);
		\draw[-latex, line width=1pt, dotted] (u1) to[out=30,in=-90] (w2);
		\draw[-latex, line width=1pt, dotted] (u2) to[out=150,in=-90] (w3);
		\draw[-latex, line width=1pt, dotted] (u2) to[out=30,in=-90] (w4);
		\draw[-latex, line width=1pt, dotted] (u3) to[out=150,in=-90] (w5);
		\draw[-latex, line width=1pt, dotted] (u3) to[out=30,in=-90] (w6);
		\draw[-latex, line width=1pt, dotted] (ux) to[out=150,in=-90] (w7);
		\draw[-latex, line width=1pt, dotted] (ux) to[out=30,in=-90] (w8);
		\draw[-latex, line width=1pt, dotted] (ux2) to[out=150,in=-90] (w9);
		\draw[-latex, line width=1pt, dotted] (ux2) to[out=30,in=-90] (w10);
		
                \draw [decoration={brace}, decorate, line width=1pt] (-0.7, 1.7) -- (0.7, 1.7) node [yshift=8pt, xshift=-20pt] {\scriptsize $d-2$};
                \draw [decoration={brace}, decorate, line width=1pt] (1.3, 1.7) -- (2.7, 1.7) node [yshift=8pt, xshift=-20pt] {\scriptsize $d-3$};
                \draw [decoration={brace}, decorate, line width=1pt] (3.3, 1.7) -- (4.7, 1.7) node [yshift=8pt, xshift=-20pt] {\scriptsize $d-3$};
                \draw [decoration={brace}, decorate, line width=1pt] (6.3, 1.7) -- (7.7, 1.7) node [yshift=8pt, xshift=-20pt] {\scriptsize $d-3$};
                \draw [decoration={brace}, decorate, line width=1pt] (8.3, 1.7) -- (9.7, 1.7) node [yshift=8pt, xshift=-20pt] {\scriptsize $d-2$};

	\end{tikzpicture}
	\caption{The worst-case situation described in the proof of Lemma~\ref{lemma:degree-to-degree}.}
	\label{figure:bounded-graphs}
\end{figure}

Whenever
a vertex $u_i$ burns at some time unit, then, at the next time unit, the fire will 
also propagate
in $\overrightarrow{G} - \overrightarrow{G'}$, that is to some vertices 
in $X$.
There are exactly 
$$d_G(u_i)-(d_{\overrightarrow{G'}}^-(u_i)+d_{\overrightarrow{G'}}^+(u_i))$$ such vertices. 
It is easy to check that the worst case, where the number of burnt vertices is at a maximum, occurs when, for $x=\overrightarrow{\beta_d}$, the path $u_1u_2\dotsm_x$ is a directed path in $\overrightarrow{G'}$, and we have $d_{\overrightarrow{G'}}^+(u_x)=1$ and $d_{\overrightarrow{G'}}^+(u_i) = 2$ for every $i \in \{1, 2, \dotsc, x-1\}$. This last condition maximises the quantity $$\sum_{i = 1}^x d_G(u_i)-(d_{\overrightarrow{G'}}^-(u_i)+d_{\overrightarrow{G'}}^+(u_i)),$$ and ensures that the firefighter cannot prevent the fire from reaching $u_x$. In particular, at each time unit~$i$, the firefighter protects a vertex $v_i$ which is an outneighbour of $u_i$ in $\overrightarrow{G'}$. See Figure~\ref{figure:bounded-graphs} for an illustration of this situation. Then, for each $u_i$ with $i \in \{2,3,\dotsc,x-1\}$, at most~$d-3$ other vertices of $\overrightarrow{G} - \overrightarrow{G'}$ will burn, and at most $d-2$ for $u_1$ and $u_x$. The total number of vertices which will burn is hence $$\overrightarrow{\beta_{d-1}} + 2 \cdot (d-2) + (\overrightarrow{\beta_{d-1}}-2)\cdot(d-3),$$ as claimed.
\end{proof}

Using the upper bound from \cref{subcubic2}, we obtain the following:

%
%

\begin{corollary} \label{corollary:better-seed}
For every $d \geq 4$, we have $\overrightarrow{\beta_d} \leq (d-1)!$.
\end{corollary}

\subsection{Planar graphs} \label{section:planar}

In this section we study the $\overrightarrow{\beta}$ parameter for planar graphs. First of all, it is well known that planar graphs have arboricity at most~$3$ (due to Schnyder), so, by \cref{corollary:f-is-delta-1,corollary:arboricity}, for every such graph $G$:

\begin{itemize}
	\item $\overrightarrow{\beta}(G, f \geq 3) = 1$, and
	\item $\overrightarrow{\beta}(G, 2) \leq \frac{|V(G)|+2}{3}$.
\end{itemize}

\noindent For this reason, we focus, in this section, on the problem of firefighting with only one firefighter in a planar graph.

\medskip

As in previous sections, our first question of interest is whether or not, for 
this family,
the $\overrightarrow{\beta}$ parameter is bounded above by an absolute constant. 

\begin{question} \label{question:planar-constant}
Is there a constant $c \geq 1$ such that $\overrightarrow{\beta}(G, 1) \leq c$ for every planar graph $G$?
\end{question}

Answering \cref{question:planar-constant} does not seem straightforward. Experiments on families of planar graphs suggest that such a constant $c$ could exist, though we are not aware of an orientation scheme and strategy that work for any planar graph.
In particular, `denser' planar graphs are problematic.  Consider the following example.
Fix a large value of $\Delta$ and let $G$ be the planar graph obtained as follows. Starting from a single vertex $v$, add a first \textit{layer} of $\Delta$ new vertices around $v$, \textit{i.e.} join $v$ to all these vertices, and add edges between the vertices of the first layer so that they induce a cycle. Now add a second layer of vertices around the first layer, and add edges between the first and second layers so that all vertices of the first layer have degree~$\Delta$. 
Repeat this procedure 
until a large number of vertices with degree $\Delta$ are obtained. Assuming $\Delta$ and the number of layers are sufficiently large, 
there is no obvious way to orient
the edges of the resulting graph 
to prevent fire propagation.

However, a simple counting argument shows that if the constant $c$ mentioned in Question~\ref{question:planar-constant} does exist, then $c \geq 3$.
A planar graph is \emph{maximal} if any graph obtained by adding an edge on the same vertex set results in a non-planar graph.
It is well known that any maximal planar graph $G$ with more than two vertices has $3|V(G)|-6$ edges.
Consequently, by applying \cref{lowerbound2}, we observe the following:

\begin{observation} \label{observation:icosahedron}
For any maximal planar graph $G$ on at least $7$ vertices, $\overrightarrow{\beta}(G, 1) \geq 3$.
\end{observation}

\noindent In fact, we will show, in \cref{observation:planar4}, that such a $c$ must be at least $4$.

\medskip


Although we have no concrete evidence that Question~\ref{question:planar-constant} has a negative answer, we suspect the following direction might be more promising.

\begin{conjecture} \label{conjecture:planar}
For every planar graph $G$, $\overrightarrow{\beta}(G, 1)$ is linear in $\Delta$.
\end{conjecture}

Since every planar graph is $4$-colourable, by the Four-Colour Theorem, it follows from \cref{proposition:bound-colorable2} that $\overrightarrow{\beta}(G, 1) \leq 2(\Delta-1)^3$ for every planar graph $G$. For some subclasses of planar graphs, this can be further improved using the wide range of 
results in the literature 
regarding these graphs.
For example, since every triangle-free planar graph is $3$-colourable by Gr\"{o}tzsch's Theorem, \cref{proposition:bound-colorable} implies that $\overrightarrow{\beta}(G, 1) \leq 2(\Delta-1)^2$ whenever $G$ is planar and triangle-free. 
Moreover, since triangle-free planar graphs have arboricity~$2$, we have, by \cref{corollary:arboricity}, that $\overrightarrow{\beta}(G,1) \leq \frac{|V(G)|+1}{2}$.

Towards Conjecture~\ref{conjecture:planar}, we now consider infinite planar grids, which received some attention for both the directed and undirected versions of the Firefighter Problem~\cite{BHW12,M08}. In particular, the strategies described below could be useful for dealing with the general case. 

We start by confirming Conjecture~\ref{conjecture:planar} for infinite rectangular grids (refer to \cref{figure:rectangular-grid} for an illustration), showing that $\Delta$ is an upper bound for $\overrightarrow{\beta}$ for these grids. 

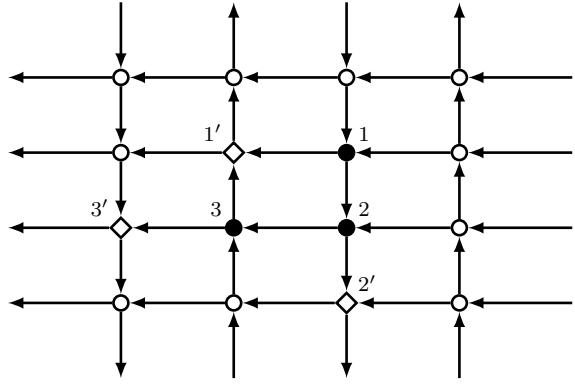
\begin{figure}[!t]
	\centering
	
	\begin{tikzpicture}[inner sep=0.7mm]
		\node[draw, circle, line width=1pt](a1) at (0,0) {};
		\node[draw, circle, line width=1pt](a2) at (1.5,0) {};
		\node[draw, circle, line width=1pt](a3) at (3,0){};
		\node[draw, circle, line width=1pt](a4) at (4.5,0) {};
		
		\node[draw, circle, line width=1pt](b1) at (0,-1){};
		\node[draw, diamond, line width=1pt](b2) at (1.5,-1)[label=above left:{\scriptsize $1'$}]{};
		\node[draw, circle, line width=1pt, fill=black](b3) at (3,-1)[label=above right:{\scriptsize $1$}]{};
		\node[draw, circle, line width=1pt](b4) at (4.5,-1) {};
		
		\node[draw, diamond, line width=1pt](c1) at (0,-2)[label=above left:{\scriptsize $3'$}] {};
		\node[draw, circle, line width=1pt, fill=black](c2) at (1.5,-2)[label=above left:{\scriptsize $3$}]{};
		\node[draw, circle, line width=1pt, fill=black](c3) at (3,-2)[label=above right:{\scriptsize $2$}]{};
		\node[draw, circle, line width=1pt](c4) at (4.5,-2){};
		
		\node[draw, circle, line width=1pt](d1) at (0,-3) {};
		\node[draw, circle, line width=1pt](d2) at (1.5,-3){};
		\node[draw, diamond, line width=1pt](d3) at (3,-3)[label=above right:{\scriptsize $2'$}] {};
		\node[draw, circle, line width=1pt](d4) at (4.5,-3) {};
				
		\draw[latex-,line width=1pt] (-1.5,0) -- (a1);
		\draw[latex-,line width=1pt] (a1) -- (a2); 
		\draw[latex-,line width=1pt] (a2) -- (a3); 
		\draw[latex-,line width=1pt] (a3) -- (a4); 
		\draw[latex-,line width=1pt] (a4) -- (6,0);
				
		\draw[latex-,line width=1pt] (-1.5,-1) -- (b1);
		\draw[latex-,line width=1pt] (b1) -- (b2); 
		\draw[latex-,line width=1pt] (b2) -- (b3); 
		\draw[latex-,line width=1pt] (b3) -- (b4); 
		\draw[latex-,line width=1pt] (b4) -- (6,-1);
				
		\draw[latex-,line width=1pt] (-1.5,-2) -- (c1);
		\draw[latex-,line width=1pt] (c1) -- (c2); 
		\draw[latex-,line width=1pt] (c2) -- (c3); 
		\draw[latex-,line width=1pt] (c3) -- (c4); 
		\draw[latex-,line width=1pt] (c4) -- (6,-2);
				
		\draw[latex-,line width=1pt] (-1.5,-3) -- (d1);
		\draw[latex-,line width=1pt] (d1) -- (d2); 
		\draw[latex-,line width=1pt] (d2) -- (d3); 
		\draw[latex-,line width=1pt] (d3) -- (d4); 
		\draw[latex-,line width=1pt] (d4) -- (6,-3);
		
		\draw[-latex,line width=1pt] (0,1) -- (a1);
		\draw[-latex,line width=1pt] (a1) -- (b1);
		\draw[-latex,line width=1pt] (b1) -- (c1);
		\draw[-latex,line width=1pt] (c1) -- (d1);
		\draw[-latex,line width=1pt] (d1) -- (0,-4);
		
		\draw[latex-,line width=1pt] (1.5,1) -- (a2);
		\draw[latex-,line width=1pt] (a2) -- (b2);
		\draw[latex-,line width=1pt] (b2) -- (c2);
		\draw[latex-,line width=1pt] (c2) -- (d2);
		\draw[latex-,line width=1pt] (d2) -- (1.5,-4);
		
		\draw[-latex,line width=1pt] (3,1) -- (a3);
		\draw[-latex,line width=1pt] (a3) -- (b3);
		\draw[-latex,line width=1pt] (b3) -- (c3);
		\draw[-latex,line width=1pt] (c3) -- (d3);
		\draw[-latex,line width=1pt] (d3) -- (3,-4);
		
		\draw[latex-,line width=1pt] (4.5,1) -- (a4);
		\draw[latex-,line width=1pt] (a4) -- (b4);
		\draw[latex-,line width=1pt] (b4) -- (c4);
		\draw[latex-,line width=1pt] (c4) -- (d4);
		\draw[latex-,line width=1pt] (d4) -- (4.5,-4);
	\end{tikzpicture}
	\caption{Firefighting in an infinite rectangular grid.}
	\label{figure:rectangular-grid}
\end{figure}

\begin{proposition}
For every infinite rectangular grid $G$, we have $\overrightarrow{\beta}(G, 1) = 3 < 4 = \Delta$.
\end{proposition}

\begin{proof}
Let $\overrightarrow{G}$ be the orientation of $G$ (depicted in \cref{figure:rectangular-grid}) obtained as follows. Orient all `rows' of $G$ from, say, `right to left'. Now, orient all `even columns' of $G$ from, say, bottom to top, and conversely for all `odd columns'. Then $\overrightarrow{G}$ is $2$-outregular, and has the property that, for every vertex $u$, one of its two outneighbours is in the third outneighbourhood of $u$. Then, when the fire starts at $u$, protecting the vertices as in \cref{figure:rectangular-grid} we can marshall the fire towards the first protected vertex, 
hence ensuring that at most $3$ vertices burn.
Since for at most three vertices to burn, we must have a $2$-outregular orientation, it is easy to check that this strategy is optimal.
\end{proof}

We now focus on infinite triangular grids (see \cref{figure:triangular-grid} for an illustration).  This case is of interest since, in order to resolve Conjecture~\ref{conjecture:planar}, 
one can restrict attention to
maximal
planar graphs, by \cref{observation:subgraph}. Here again, we confirm that $\Delta$ is an upper bound for $\overrightarrow{\beta}$.

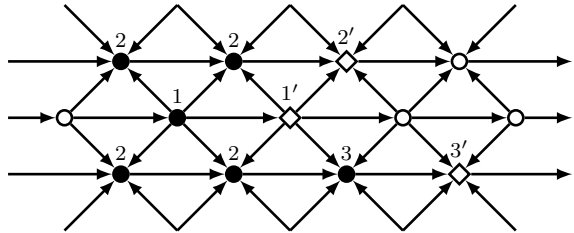
\begin{figure}[!t]
	\centering
	
	\begin{tikzpicture}[inner sep=0.7mm]
		\node[draw, circle, line width=1pt, fill=black](a1) at (0,0)[label=above:{\scriptsize $2$}]{};
		\node[draw, circle, line width=1pt, fill=black](a2) at (1.5,0)[label=above:{\scriptsize $2$}]{};
		\node[draw, diamond, line width=1pt](a3) at (3,0)[label=above:{\scriptsize $2'$}]{};
		\node[draw, circle, line width=1pt](a4) at (4.5,0) {};
		
		\node[draw, circle, line width=1pt](b0) at (-0.75,-0.75) {};
		\node[draw, circle, line width=1pt, fill=black](b1) at (0.75,-0.75)[label=above:{\scriptsize $1$}]{};
		\node[draw, diamond, line width=1pt](b2) at (2.25,-0.75)[label=above:{\scriptsize $1'$}]{};
		\node[draw, circle, line width=1pt](b3) at (3.75,-0.75) {};
		\node[draw, circle, line width=1pt](b4) at (5.25,-0.75) {};
		
		\node[draw, circle, line width=1pt, fill=black](c1) at (0,-1.5)[label=above:{\scriptsize $2$}]{};
		\node[draw, circle, line width=1pt, fill=black](c2) at (1.5,-1.5)[label=above:{\scriptsize $2$}]{};
		\node[draw, circle, line width=1pt, fill=black](c3) at (3,-1.5)[label=above:{\scriptsize $3$}]{};
		\node[draw, diamond, line width=1pt](c4) at (4.5,-1.5)[label=above:{\scriptsize $3'$}]{};
		
		\draw[-latex,line width=1pt] (-0.75,0.75) -- (a1);
		\draw[-latex,line width=1pt] (0.75,0.75) -- (a1);
		\draw[-latex,line width=1pt] (0.75,0.75) -- (a2);
		\draw[-latex,line width=1pt] (2.25,0.75) -- (a2);
		\draw[-latex,line width=1pt] (2.25,0.75) -- (a3);
		\draw[-latex,line width=1pt] (3.75,0.75) -- (a3);
		\draw[-latex,line width=1pt] (3.75,0.75) -- (a4);
		\draw[-latex,line width=1pt] (5.25,0.75) -- (a4);
		
		\draw[-latex,line width=1pt] (-1.5,0) -- (a1);
		\draw[-latex,line width=1pt] (a1) -- (a2);
		\draw[-latex,line width=1pt] (a2) -- (a3);
		\draw[-latex,line width=1pt] (a3) -- (a4);
		\draw[-latex,line width=1pt] (a4) -- (6,0);
		
		\draw[-latex,line width=1pt] (b0) -- (a1);
		\draw[-latex,line width=1pt] (b0) -- (c1);

		\draw[-latex,line width=1pt] (-1.5,-0.75) -- (b0);		
		
		\draw[-latex,line width=1pt] (b0) -- (b1);
		\draw[-latex,line width=1pt] (b1) -- (b2);
		\draw[-latex,line width=1pt] (b2) -- (b3);
		\draw[-latex,line width=1pt] (b3) -- (b4);
		\draw[-latex,line width=1pt] (b4) -- (6,-0.75);
		
		\draw[-latex,line width=1pt] (b4) -- (a4);
		\draw[-latex,line width=1pt] (b4) -- (c4);
		
		\draw[-latex,line width=1pt] (-1.5,-1.5) -- (c1);
		\draw[-latex,line width=1pt] (c1) -- (c2);
		\draw[-latex,line width=1pt] (c2) -- (c3);
		\draw[-latex,line width=1pt] (c3) -- (c4);
		\draw[-latex,line width=1pt] (c4) -- (6,-1.5);
		
		\draw[-latex,line width=1pt] (b1) -- (a1);
		\draw[-latex,line width=1pt] (b1) -- (a2);
		\draw[-latex,line width=1pt] (b2) -- (a2);
		\draw[-latex,line width=1pt] (b2) -- (a3);
		\draw[-latex,line width=1pt] (b3) -- (a3);
		\draw[-latex,line width=1pt] (b3) -- (a4);
		
		\draw[-latex,line width=1pt] (b1) -- (c1);
		\draw[-latex,line width=1pt] (b1) -- (c2);
		\draw[-latex,line width=1pt] (b2) -- (c2);
		\draw[-latex,line width=1pt] (b2) -- (c3);
		\draw[-latex,line width=1pt] (b3) -- (c3);
		\draw[-latex,line width=1pt] (b3) -- (c4);
		
		\draw[-latex,line width=1pt] (-0.75,-2.25) -- (c1);
		\draw[-latex,line width=1pt] (0.75,-2.25) -- (c1);
		\draw[-latex,line width=1pt] (0.75,-2.25) -- (c2);
		\draw[-latex,line width=1pt] (2.25,-2.25) -- (c2);
		\draw[-latex,line width=1pt] (2.25,-2.25) -- (c3);
		\draw[-latex,line width=1pt] (3.75,-2.25) -- (c3);
		\draw[-latex,line width=1pt] (3.75,-2.25) -- (c4);
		\draw[-latex,line width=1pt] (5.25,-2.25) -- (c4);
	\end{tikzpicture}
	\caption{Firefighting in an infinite triangular grid.}
	\label{figure:triangular-grid}
\end{figure}

\begin{proposition} \label{proposition:planar4}
For every infinite triangular grid $G$, we have $\overrightarrow{\beta}(G, 1) \leq 6 = \Delta$.
\end{proposition}

\begin{proof}
Let $\overrightarrow{G}$ be an orientation of $G$ as depicted in \cref{figure:triangular-grid}. Namely, the vertices of $G$ are decomposed into several layers, \textit{i.e.} `parallel' chains of consecutive adjacent vertices. All these layers are oriented in the same direction. Finally, the edges between two consecutive layers are oriented so that all `even layers' are, say, `sinks' (\textit{i.e.} have all their incident arcs incoming) while all `odd layers' are `sources' (\textit{i.e.} have all their incident arcs outgoing).

Assume the fire starts at some vertex $u$. If $u$ belongs to a sink layer, then $u$ has outdegree~$1$ so the firefighter can just contain the fire by protecting the outneighbour of $u$ at time~$1$. Now, if $u$ belongs to a source layer, then just apply the strategy described in \cref{figure:triangular-grid}, consisting in first protecting the layer of $u$, and then successively protecting the two adjacent sink layers. From this, we deduce that we can ensure that at most $6$ vertices have burnt by the time the fire is contained.
\end{proof}

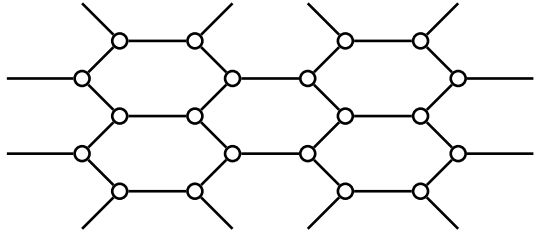
\begin{figure}[!t]
	\centering
	
	\begin{tikzpicture}[inner sep=0.7mm]
		
		\node[draw, circle, line width=1pt](a1) at (0,0) {};
		\node[draw, circle, line width=1pt](a2) at (1,0) {};
		\node[draw, circle, line width=1pt](a3) at (3,0) {};
		\node[draw, circle, line width=1pt](a4) at (4,0) {};
		
		\draw[-,line width=1pt] (a2) -- (1.5,0.5);
		\draw[-,line width=1pt] (a1) -- (-0.5,0.5);
		\draw[-,line width=1pt] (a3) -- (2.5,0.5);
		\draw[-,line width=1pt] (a4) -- (4.5,0.5);
		
		\node[draw, circle, line width=1pt](b1) at (-0.5,-0.5) {};
		\node[draw, circle, line width=1pt](b2) at (1.5,-0.5) {};
		\node[draw, circle, line width=1pt](b3) at (2.5,-0.5) {};
		\node[draw, circle, line width=1pt](b4) at (4.5,-0.5) {};
		
		\draw[-,line width=1pt] (b1) -- (-1.5,-0.5);
		\draw[-,line width=1pt] (b4) -- (5.5,-0.5);
		
		\node[draw, circle, line width=1pt](c1) at (0,-1) {};
		\node[draw, circle, line width=1pt](c2) at (1,-1) {};
		\node[draw, circle, line width=1pt](c3) at (3,-1) {};
		\node[draw, circle, line width=1pt](c4) at (4,-1) {};
		
		\node[draw, circle, line width=1pt](d1) at (-0.5,-1.5) {};
		\node[draw, circle, line width=1pt](d2) at (1.5,-1.5) {};
		\node[draw, circle, line width=1pt](d3) at (2.5,-1.5) {};
		\node[draw, circle, line width=1pt](d4) at (4.5,-1.5) {};
		
		\draw[-,line width=1pt] (d1) -- (-1.5,-1.5);
		\draw[-,line width=1pt] (d4) -- (5.5,-1.5);
		
		\node[draw, circle, line width=1pt](e1) at (0,-2) {};
		\node[draw, circle, line width=1pt](e2) at (1,-2) {};
		\node[draw, circle, line width=1pt](e3) at (3,-2) {};
		\node[draw, circle, line width=1pt](e4) at (4,-2) {};
		
		\draw[-,line width=1pt] (e2) -- (1.5,-2.5);
		\draw[-,line width=1pt] (e1) -- (-0.5,-2.5);
		\draw[-,line width=1pt] (e3) -- (2.5,-2.5);
		\draw[-,line width=1pt] (e4) -- (4.5,-2.5);
		
		\draw[-,line width=1pt] (a1) -- (a2);
		\draw[-,line width=1pt] (a3) -- (a4);
		\draw[-,line width=1pt] (b1) -- (a1);
		\draw[-,line width=1pt] (b1) -- (c1);
		\draw[-,line width=1pt] (b2) -- (a2);
		\draw[-,line width=1pt] (b2) -- (b3);
		\draw[-,line width=1pt] (b3) -- (a3);
		\draw[-,line width=1pt] (a4) -- (b4);
		\draw[-,line width=1pt] (c1) -- (c2);
		\draw[-,line width=1pt] (c3) -- (c4);
		\draw[-,line width=1pt] (c2) -- (b2);
		\draw[-,line width=1pt] (c3) -- (b3);
		\draw[-,line width=1pt] (c4) -- (b4);
		\draw[-,line width=1pt] (d1) -- (c1);
		\draw[-,line width=1pt] (d2) -- (c2);
		\draw[-,line width=1pt] (d2) -- (d3);
		\draw[-,line width=1pt] (d3) -- (c3);
		\draw[-,line width=1pt] (d4) -- (c4);
		\draw[-,line width=1pt] (e1) -- (d1);
		\draw[-,line width=1pt] (e2) -- (d2);
		\draw[-,line width=1pt] (e3) -- (d3);
		\draw[-,line width=1pt] (e4) -- (d4);
		\draw[-,line width=1pt] (e1) -- (e2);
		\draw[-,line width=1pt] (e3) -- (e4);
	\end{tikzpicture}
	\caption{(Part of) an infinite hexagonal grid.}
	\label{figure:hexagonal-grid}
\end{figure}

\begin{observation}\label{observation:planar4}
  There exists a planar graph $G$ with $$\overrightarrow{\beta}(G, 1) \geq 4.$$
\end{observation}
\begin{proof}
  We show that at least four vertices burn no matter how we orient a sufficiently large triangulated grid. For such a graph, assume, towards a contradiction, there is an orientation $\overrightarrow{G}$ by which at most three vertices can burn when firefighting with only one firefighter. Then the maximum outdegree of $\overrightarrow{G}$ is at most $3$ and $\overrightarrow{G}$ is locally $3$-outregular. Assuming the fire starts at some vertex $u$, we can protect one of the three outneighbours of $u$ at time~$1$ before the fire propagates to two new vertices $v_1$ and $v_2$. Now the outneighbourhood of $v_1$ and $v_2$ must be of size at most~$1$, which is impossible due to the structure of $\overrightarrow{G}$ and the fact that $v_1$ and $v_2$ have outdegree~$3$; a contradiction. So a fourth vertex must burn.
\end{proof}


We finish this section by remarking that for infinite hexagonal grids (as depicted in \cref{figure:hexagonal-grid}), even more vertices can be saved: namely all but at most~$2$. This follows from \cref{subcubic2}, since hexagonal grids are subcubic graphs.

\section{Characterising graphs by the number of vertices that burn} \label{section:beta1}

In this section we consider the problem of characterising the class of graphs for which at most $k$ vertices burn using an optimal firefighting strategy.  That is, we wish to determine the class of graphs
%
$$\bk{k} = \{G : \overrightarrow{\beta}(G, 1) \leq k\},$$
for a positive integer $k$.
Note that, by definition, we have $\bk{1} \subseteq \bk{2} \subseteq \dotsm \subseteq \bk{\infty}$, where $\bk{\infty}$ is the class of all graphs.
Although such a characterisation may be difficult in general, 
we give
an explicit characterisation of $\bk{1}$,
discuss what we know about $\bk{2}$, 
and give some necessary conditions for membership in $\bk{k}$.

\begin{theorem} \label{theorem:charac-b1}
Let $G$ be a connected graph. Then $G \in \bk{1}$ if and only if $G$ contains at most one cycle.
\end{theorem}

\begin{proof}
  ($\Leftarrow$)
If $G$ has no cycles, it is a tree, so $\overrightarrow{\beta}(G, 1) = 1$ by \cref{proposition:tree}. Now, if $G$ is unicyclic, we can start by orienting its unique cycle $C$ 
such
that each of its vertices has outdegree~$1$ 
in $C$.
Then, for every 
component $T$ of $G \backslash E(C)$,
orient its edges from the leaves towards $C$. Then the outdegree of every vertex in 
$V(T) \setminus V(C)$
is exactly~$1$, while the outdegrees of the vertices in $C$ have not changed. The resulting orientation is therefore $1$-outregular, so the fire can be immediately blocked at time~$1$.

  ($\Rightarrow$)
  Suppose $G$ is has distinct cycles $C_1$ and $C_2$.  After removing an edge in $E(C_1) \setminus E(C_2)$ from $G$, the resulting graph contains the cycle $C_2$, so is not a tree.  Hence $G$ has more than $|V(G)|$ edges.
  By \cref{lowerbound2}, $\overrightarrow{\beta}(G, 1) > 1$, so $G$ is not in $\bk{1}$.
\end{proof}

%
Now we consider $\bk{2}$.
By earlier results, this class contains all cubic graphs, $K_5$, complete bipartite graphs of the form $K_{2,n}$, all partial $2$-trees (thus, series-parallel graphs and outerplanar graphs), and all subgraphs thereof.
On the other hand, it does not contain the entire class of planar graphs; in particular, it does not contain any maximal planar graph with at least seven vertices.
%
It also does not contain all graphs with maximum degree~$4$ (for example $K_{4,4}$).  However, it does contain arbitrarily large $4$-regular graphs; one example of such a graph is given in \cref{figure:diamondcycle}.  However, it can be shown that every $4$-regular graph in $\bk{2}$ has particular structure; namely, every vertex is in a \emph{diamond} (a graph that can be obtained by removing an edge from $K_4$).

\begin{figure}
  \centering
  \begin{tikzpicture}[inner sep=0.7mm]
    \node[draw, circle, line width=1pt](a1) at (1,0) {};
    \node[draw, diamond, line width=1pt](a2) at (0,1)[label=above:{\scriptsize $2'$}] {};
    \node[draw, circle, line width=1pt, fill=black](a3) at (-1,0)[label=left:{\scriptsize $2$}] {};
    \node[draw, circle, line width=1pt](a4) at (0,-1) {};
    
    \node[draw, circle, line width=1pt](b1) at (2,2) {};
    \node[draw, diamond, line width=1pt](b2) at (-2,2)[label=left:{\scriptsize $1'$}] {};
    \node[draw, circle, line width=1pt, fill=black](b3) at (-2,-2)[label=left:{\scriptsize $1$}] {};
    \node[draw, circle, line width=1pt](b4) at (2,-2) {};
    
    \draw[-latex,line width=1pt] (a1) -- (a4);
    \draw[-latex,line width=1pt] (a2) -- (a1);%
    \draw[-latex,line width=1pt] (a3) -- (a2);
    \draw[-latex,line width=1pt] (a4) -- (a3);

    \draw[-latex,line width=1pt] (b1) -- (b4);%
    \draw[-latex,line width=1pt] (b2) -- (b1);
    \draw[-latex,line width=1pt] (b3) -- (b2);
    \draw[-latex,line width=1pt] (b4) -- (b3);

    \draw[-latex,line width=1pt] (b1) -- (a1);%
    \draw[-latex,line width=1pt] (b2) -- (a2);
    \draw[-latex,line width=1pt] (b3) -- (a3);
    \draw[-latex,line width=1pt] (b4) -- (a4);

    \draw[-latex,line width=1pt] (a2) -- (b1);%
    \draw[-latex,line width=1pt] (a3) -- (b2);
    \draw[-latex,line width=1pt] (a4) -- (b3);
    \draw[-latex,line width=1pt] (a1) -- (b4);
  \end{tikzpicture}
  \caption{A $4$-regular graph contained in $\protect\bk{2}$, and an optimal orientation.}
  \label{figure:diamondcycle}
\end{figure}
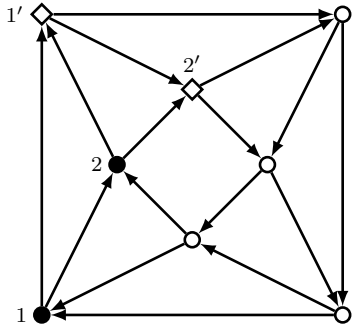

For a graph $G$ to be a member of the class, it is necessary
that, for every subgraph $G'$ of $G$, we have $|E(G')| \leq 2|V(G')|$, by \cref{lowerbound2,observation:subgraph}.
This is not a sufficient condition, however; for example, $K_{4,4}$ satisfies the condition but $\overrightarrow{\beta}(K_{4,4},1) \geq 3$ by \cref{proposition:lower-bipartite}.

We can also deduce a necessary condition in terms of the minimum degree of $G$.
If $G$ has minimum degree $\delta$, then, by the handshaking lemma, $|E(G)| \geq \frac{\delta}{2}|V(G)|$.  Hence, by \cref{lowerbound2}, we have $\overrightarrow{\beta}(G,1) \geq \frac{\delta}{2}$.
Thus, another necessary condition for membership in $\bk{2}$ is that the graph has minimum degree at most $4$.  By \cref{observation:subgraph}, any subgraph must also have this property, which implies that, moreover, it is necessary that the graph is $4$-degenerate.
Again, $K_{4,4}$ is an example that demonstrates these conditions are not sufficient.

\begin{problem}
Characterise $\bk{2}$.
\end{problem}

More generally, we can deduce necessary conditions for a graph $G$ to be in $\bk{k}$.
Namely, if $G$ is a member of $\bk{k}$, then
every subgraph $G'$ of $G$ satisfies $|E(G')| \leq k |V(G')|$. 
%
Moreover,
it is necessary that a graph $G$ in $\bk{k}$ is $2k$-degenerate and, in particular, has minimum degree at most $2k$.  

\bigskip

\noindent \textbf{Acknowledgements:} The authors would like to thank Prof.\ Gary MacGillivray for his talk on the Firefighter Problem at BGW14, which inspired this work.

\section*{References}

\bibliographystyle{plain}
\bibliography{ref}

\end{document}